\documentclass[accepted=2026-03-02,a4paper,11pt,onecolumn]{quantumarticle}
\pdfoutput=1

\usepackage[T1]{fontenc}
\usepackage[english]{babel}
\usepackage{csquotes}
\usepackage{amsmath,amssymb,amsfonts}
\usepackage{amsthm}
\usepackage[dvipsnames]{xcolor}
\usepackage{bm}
\usepackage{graphicx}
\usepackage{thmtools}
\usepackage{enumerate}
\usepackage{relsize}
\usepackage{subcaption}
\usepackage[breaklinks=true]{hyperref}
\usepackage[ruled,vlined]{algorithm2e}
\usepackage[capitalise,nameinlink]{cleveref}
\usepackage{anyfontsize}
\usepackage{microtype} 
\usepackage{cs270}
\usepackage{braket}
\usepackage{tikz}
\usetikzlibrary{calc, fit}

\usepackage[arxiv=abs,eprint=true,articletitle=true,biblabel=brackets, backend=biber, style=phys,doi=false]{biblatex}
\DeclareFieldFormat{titlecase}{#1}

\begin{document}

\title{Optimizing Sparse SYK}
\author[1,2,3]{Matthew Ding}
\email{mcding@stanford.edu}
\orcid{0000-0001-7674-0548}

\author[4,3,2]{Robbie King}
\orcid{0000-0002-8152-6340}

\author[5,6]{Bobak T.\ Kiani}
\orcid{0000-0003-1477-0308}

\author[2]{Eric R.\ Anschuetz}
\orcid{0000-0002-9825-3692}

\affil[1]{Stanford University, Stanford, CA 94305, USA}
\affil[2]{California Institute of Technology, Pasadena, CA 91125, USA}
\affil[3]{University of California, Berkeley, CA 94720, USA}
\affil[4]{Google Quantum AI, Venice, CA 90291, USA}
\affil[5]{Bowdoin College, Brunswick, ME 04011, USA}
\affil[6]{Harvard University, Cambridge, MA 02138, USA}

\begin{abstract}
    Finding the ground state of strongly-interacting fermionic systems is often the prerequisite for fully understanding both quantum chemistry and condensed matter systems. The Sachdev--Ye--Kitaev (SYK) model is a representative example of such a system; it is particularly interesting not only due to the existence of efficient quantum algorithms preparing approximations to the ground state such as Hastings--O'Donnell (STOC 2022), but also known no-go results for many classical ansatzes in preparing low-energy states. However, this quantum-classical separation is known to \emph{not} persist when the SYK model is sufficiently sparsified, i.e., when terms in the model are discarded with probability $1-p$, where $p=\Theta(1/n^3)$ and $n$ is the system size. This raises the question of how robust the quantum and classical complexities of the SYK model are to sparsification.
    
    In this work we initiate the study of the sparse SYK model where $p \in [\Theta(1/n^3),1]$ and show there indeed exists a certain robustness of sparsification. We prove that with high probability, Gaussian states achieve only a $\Theta(1/\sqrt{n})$-factor approximation to the true ground state energy of sparse SYK for all $p\geq\Omega(\log n/n^2)$, and that Gaussian states cannot achieve constant-factor approximations unless $p \leq O(\log^2 n/n^3)$. Additionally, we prove that the quantum algorithm of Hastings--O'Donnell still achieves a constant-factor approximation to the ground state energy when $p\geq\Omega(\log n/n)$. Combined, these show a provable separation between classical algorithms outputting Gaussian states and efficient quantum algorithms for the goal of finding approximate sparse SYK ground states whenever $p \geq \Omega(\log n/n)$, extending the analogous $p=1$ result of Hastings--O'Donnell.
\end{abstract}
\maketitle

\newpage
\tableofcontents
\newpage

\section{Introduction}\label{sec:intro}

Simulating complex quantum mechanical systems is a potential key use-case for quantum computers \cite{Feynman_1982}, and for this reason understanding which classes of systems are ``quantumly easy'' yet ``classically hard'' to simulate is important in determining practical applications for future quantum devices. A particularly exciting set of physically relevant systems are quantum chemistry Hamiltonians, whose physical properties are often difficult to compute using classical algorithms~\cite{cao2019quantum}. These Hamiltonians are composed of low-weight \emph{fermionic} interactions \cite{BRAVYI2002210} in contrast to the setting of \emph{spin} Hamiltonians more commonly studied in Hamiltonian complexity theory~\cite{aharonov2002quantum,kempe2006complexity, Gharibian_2015}.

One central task in computing the properties of quantum chemistry systems is finding the ground state~\cite{cao2019quantum}. Though this is difficult in general---even for quantum computers---there exist classical mean-field methods such as Hartree--Fock which output \emph{Gaussian state} approximations to the ground states of such systems. Understanding (I) when Gaussian states do not well-approximate the true ground state energy and (II) when quantum algorithms do is thus important for understanding where one might find a quantum advantage in a practically relevant task.

It was found in the specific case of \emph{strongly-interacting} fermionic Hamiltonians that both of these properties hold for typical instances \cite{hastings2023optimizing, anschuetz2024strongly}, i.e., (I) Gaussian states provably fail and (II) quantum algorithms succeed on typical instances. A popular model for the strongly-interacting regime is the Sachdev--Ye--Kitaev (SYK) model \cite{SachdevYeSYK, KitaevSYK}, representing all-to-all degree-4 fermionic interactions
\begin{equation}
    H_{SYK} \triangleq \binom{2n}{4}^{-1/2}\sum_{\substack{I\subset[2n]\\|I|=4}} J_I \gamma_I,
\end{equation}
where each $J_I$ is an i.i.d.\ standard Gaussian and $\gamma_I$ is the Hermitian monomial formed by the Majorana operators indexed by the set $I$. SYK is an important fermionic Hamiltonian model due to its applications to various physical phenomena, including electron correlations in material science \cite{Jiang_2018} and the holographic AdS/CFT correspondence in quantum gravity \cite{Rosenhaus_2019}. 

Physics-based heuristics originally showed that the ground state energy of the SYK model is $-\Theta(\sqrt{n})$ (e.g., \cite{Garcia-Garcia2018}); however, this was only rigorously proven by the two combined works of Feng--Tian--Wei \cite{feng2019spectrum} and Hastings--O'Donnell \cite{hastings2023optimizing}. Feng, Tian, and Wei showed that the SYK Hamiltonian has an expected maximum eigenvalue\footnote{Equivalent to the ground state energy up to a sign as $H_{SYK}$ and $-H_{SYK}$ are distributed identically.} upper-bounded by $\sqrt{\ln 2}\sqrt{n}$. Hastings and O'Donnell showed the corresponding lower bound on the maximum eigenvalue of $\Omega(\sqrt{n})$ (with high probability) by constructing an efficient quantum algorithm which certifies this lower bound for specific SYK instances. 

On the other hand, Haldar--Tavakol--Scaffidi \cite{HTS21} showed that Gaussian states can achieve only $O(1)$ energy on typical SYK Hamiltonians. Combined with the lower bound by Hastings and O'Donnell \cite{hastings2023optimizing}, this disproves the possibility of Gaussian states achieving any constant-fraction approximation of the maximal energy. \color{black} This Gaussian state energy upper bound was extended in recent work by Anschuetz, Chen, Kiani, and King \cite{anschuetz2024strongly}, who showed that any high-energy state of the SYK model must have high circuit complexity. This demonstrates that any classical ansatz (whose energy can be computed on a classical computer) with an efficient representation---not just Gaussian states---cannot be close to the maximum energy of the SYK model. These results suggest that simulating the SYK model is a strong candidate for showcasing a quantum advantage.

Despite this potential for quantum advantage, it remains a major engineering challenge to simulate a Hamiltonian with $\Theta(n^4)$ degree-4 interactions on quantum hardware. To find systems which can possibly be realized by quantum computers in the near future, it is advantageous to analyze strongly-interacting Hamiltonians with \emph{sparse} interactions where each particle only interacts with a few others in total. The \emph{sparse SYK model} was proposed by \cite{sparse_holography} as a modified version of SYK where each interaction term is removed independently with probability $1-p$:
\begin{equation}\label{eqn:SSYK}
     H_{SSYK}(p) \triangleq \binom{2n}{4}^{-1/2}p^{-1/2}\sum_{\substack{I\subset[2n]\\|I|=4}} J_I X_I \gamma_I.
\end{equation}
Here, $X_I$ are i.i.d.\ Bernoulli random variables with $\Pr(X_I=1)=p$.

Determining whether one can simulate the interesting physics of strongly-interacting Hamiltonians on near-term quantum devices using sparsified models is still unclear. Recent work attempted to simulate wormhole dynamics on a quantum computer through a specific sparse SYK model generated through machine learning \cite{jafferis2022traversable}; however, it is an open question whether this specific sparsification actually exhibited holographic behavior~\cite{kobrin2023commenttraversablewormholedynamics,orman2025quantum}. Additionally, Herasymenko, Stroeks, Helsen, and Terhal \cite{Herasymenko_2023} showed that a Gaussian-approximation gap does \emph{not} exist when the SYK model is too sparse. They provide an efficient classical algorithm which outputs a Gaussian state with a constant-factor approximation to the sparse SYK ground state when $p=\Theta\left(1/n^3\right)$. This and other previous works on sparse SYK \cite{sparse_holography, Herasymenko_2023,Hothem2023, orman2025quantum} have only considered in detail the case where $p=\Theta(1/n^3)$, i.e., when each Majorana has in expectation a constant number of interactions. These results leave open the question of what happens between the two extremes of $p=1$ and $p=\Theta(1/n^3)$:
\begin{center}
    \textbf{How does sparsity in strongly-interacting fermionic Hamiltonians affect the approximation gap between classical and quantum algorithms?}
\end{center}
We make progress towards further answering this question in two ways:
\begin{enumerate}
    \item (\textbf{Classical hardness}, \cref{sec:gaussian-upper}) We show that Gaussian states achieve only a constant $O(1)$ energy on the sparse SYK model for all $p > \widetilde{\Omega}(1/n^2)$. We additionally show that only an asymptotically-vanishing (sub-constant) approximation ratio is achievable by (I) any Gaussian state when $p > \widetilde{\Omega}(1/n^3)$, (II) general states with circuit-complexity similar to Gaussian states when $p > \widetilde{\Omega}(1/n)$, and (III) various low-magic ansatzes when $p > \widetilde{\Omega}(1/n^2)$.
    \item (\textbf{Quantum easiness}, \cref{sec:sparse-ho-alg}) We prove a specific quantum algorithm achieves a constant approximation ratio for the sparse SYK model when $p > \widetilde{\Omega}(1/n)$. 
\end{enumerate}

This is supplemented by our additional ``universality'' result in \cref{sec:universality}, which states that the maximum eigenvalue for sparse SYK is always $\Theta(\sqrt{n})$ with high probability, regardless of the value of $p\in\left[\Theta(1/n^3),1\right]$. The maximum energy achievable by any quantum state for all sparse SYK or SYK Hamiltonians is therefore, up to constant factors, invariant to sparsification with high probability. We summarize our results in \cref{fig:energy-diagram}.

Our results prove that the sparse SYK model retains the maximum $1/\sqrt{n}$ Gaussian approximation gap for all $p \gtrsim 1/n^2$, suggesting similar physical properties between SYK and sparse SYK in this sparsity range. One such example is the fact that low-energy states of the SYK model have large amounts of magic (see \cite{cudby2025magicstates, PhysRevLett.123.080503} for details on fermionic magic states); this property was known for the fully dense SYK model~\cite{hastings2023fieldtheory}, and we show in \Cref{cor:magic} that this property is robust to sparsifications through $p \gtrsim 1/n^2$. Our work also shows that ``reasonable'' properties of physical Hamiltonians like sparsity do not necessarily preclude classical hardness. This suggests the existence of other sparse and physically-relevant Hamiltonians where quantum computers can have provable advantages over classical ones. Finally, our results prove that the algorithms producing constant-approximation ratio Gaussian states when $p=\Theta(1/n^3)$ \cite{Herasymenko_2023, Hothem2023} fail to for any slower decay of $p$ with $n$ (up to logarithmic factors).

\begin{figure}[t]
    \centering
    \includegraphics[width=\textwidth]{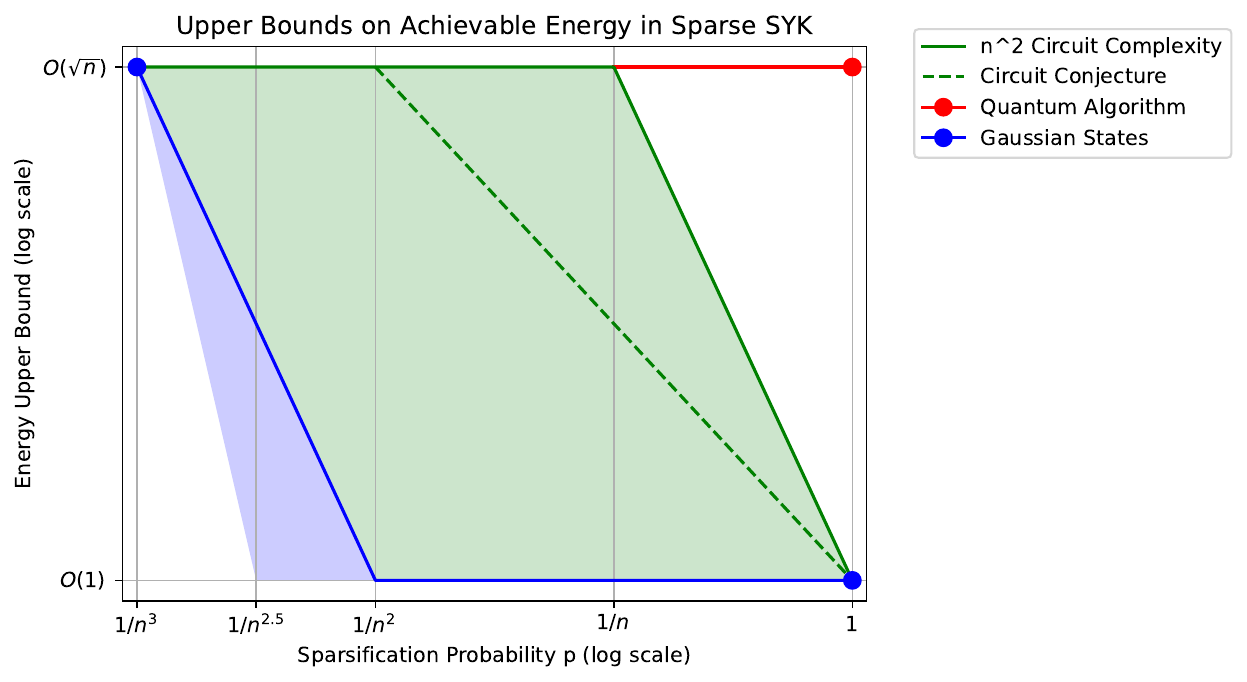}
    \caption{We plot the maximum achievable energies on sparse SYK Hamiltonians for various classes of quantum states (up to logarithmic factors in $p$). Green represents arbitrary disorder-independent sets of $\exp(n^2\log n)$ states, including states of circuit complexity $\operatorname{\widetilde{\Theta}}\left(n^2\right)$; blue represents Gaussian states; and red represents general quantum states generated from known quantum algorithms. Solid lines represent upper bounds proved in this work, dots represent bounds shown in previous work \cite{HTS21,hastings2023optimizing, Herasymenko_2023}, dashed lines denote conjectured upper bounds in this work, and shaded regions represent where the true energy achieved by the state class may lie, i.e., where there is no tight matching lower bound.}
    \label{fig:energy-diagram}
\end{figure}

\subsection{Gaussian state upper bound}
We show strong upper bounds on the optimal energy Gaussian states can achieve on sparse SYK Hamiltonians which interpolates between the $\Theta(1)$ maximum achievable energy at $p=1$ \cite{HTS21, hastings2023optimizing} and the $\Theta(\sqrt{n})$ maximum achievable energy at $p=\Theta(1/n^3)$ \cite{Herasymenko_2023, Hothem2023}.

\begin{restatable}[Gaussian state energy upper bound]{theorem}{gaussianthm}\label{thm:gaussian-upper-bound}
    In the sparse SYK model with sparsity parameter $p$, with high probability over the draw of the Hamiltonian, the set of all Gaussian states achieves a maximum energy upper-bounded by $O(1)$ when $p \geq \Omega(\log n/n^2)$ and a maximum energy upper-bounded by $O(\frac{\log n}{n\sqrt{p}})$ when $p < o(\log n/n^2)$.
\end{restatable}

This upper-bound is tight for the sparsity range $\Omega(\log n/n^2) \leq p \leq 1$ due to the existence of Gaussian states with energy $\Omega(1)$ for all $p$ (\cref{sec:gaussian_lower}). We prove \cref{thm:gaussian-upper-bound} in \cref{subsec:gaussian-energy-upper} by adapting a tensor network construction from \cite{hastings2023fieldtheory} originally used to prove upper bounds on ansatzes based on Gaussian states for the standard SYK model. A similar method was originally introduced in \cite{HTS21} to prove the $O(1)$ energy upper bound for Gaussian states on SYK. By expanding the energy calculation for an arbitrary Gaussian state utilizing Wick's theorem, we are able to represent the value as a specific tensor network contraction. Using bounds on the operator norm of Gaussian matrices, the works of \cite{HTS21, hastings2023fieldtheory} are able to provide energy upper bounds for arbitrary Gaussian states and specific fermionic ansatzes with low magic. In our setting, the tensors are now formed with Gaussian-times-Bernoulli random variables; we apply a stronger operator norm bound for the sparse SYK setting that allows us to achieve similar results for a large range of sparsities. As a direct corollary also due to \cite{hastings2023fieldtheory}, our results directly give energy upper bounds on various additional ansatzes which are based on Gaussian states for sparse SYK (\cref{appendix:entangled-gaussian}), including fermionic states with low magic.

In addition, we also show any fixed (disorder-independent) collection of states $S$ containing a state achieving energy at least $t$ (with high probability over the disorder) on the sparse SYK model must have at least $|S| \geq \exp(\Omega(pn^2t^2))$ distinct states (\cref{subsec:circuit-complexity-upper}). This applies a separate upper-bound on the energy achievable by any ansatz constructed with small circuits. Gaussian states have circuit complexity $\sim n^2$, so this gives a (comparatively weaker to \cref{thm:gaussian-upper-bound}) energy upper-bound of $\widetilde{O}(p^{1/2})$ which only remains nontrivial when $p\gtrsim 1/n$. We also conjecture the stronger bound $|S| \geq \exp(\Omega(\sqrt{p}n^2t^2))$ for achieving energy at least $t$, supported with numerical simulations.

The circuit complexity results are achieved with the following framework: (I) We show that there exists a finite set of states which forms an $\epsilon$-net of all ansatz states in terms of energy with respect to typical sparse SYK Hamiltonians. (II) We use a Lovász theta function on graphs to bound the maximum variance of any fixed quantum state with respect to the distribution of sparse SYK Hamiltonians when $p \geq \Omega(\log n/n)$. A similar argument was shown in \cite{anschuetz2024strongly} for the original SYK model. (III) With a tight enough bound on this variance, we can use standard concentration bounds to show that the supremum of the energy over all ansatz states in the net strongly concentrates around its mean, and therefore cannot be too large. Since the supremum of the net well-approximates the supremum of all ansatz states, with high probability no ansatz state exists with high energy for typical sparse SYK instances.

\subsection{Quantum algorithm}
We additionally show that a certain quantum algorithm outputs an explicit quantum state which achieves a constant approximation of the ground state energy in sparse SYK Hamiltonians when $p \geq \Omega(\log n/n)$.

\begin{restatable}[Sparse SYK quantum algorithm]{theorem}{quantumthm}\label{thm:sparse-ho}
    In the sparse SYK model with any sparsity parameter $p \geq \Omega(\log n/n)$, there exists an efficient quantum algorithm outputting a state which achieves $\Omega(\sqrt{n})$ energy with probability exponentially close to $1$ over the draw of the Hamiltonian.
\end{restatable}

This quantum algorithm is the Hastings--O'Donnell variational algorithm \cite{hastings2023optimizing} for optimizing the SYK model, which works as follows:

\paragraph{Hastings--O'Donnell algorithm overview.}

Hastings--O'Donnell defines a ``two-color'' SYK model, which is the Hamiltonian
\begin{equation}
    H^{(2)}_{SYK}(n_1, n_2) \triangleq \frac{i}{\sqrt{n_2}}\sum_{m=1}^{n_2} \tau_m \gamma_m
\end{equation}
where $\tau_m$ is an instance of a $3$-body SYK Hamiltonian on $n_1$ Majorana modes. One can think of this Hamiltonian as partitioning the standard SYK Hamiltonian operators into two sets, $A$ of size $3n/4$ and $B$ of size $n/4$, and only considering interaction terms which have three terms from $A$ and one term from $B$. This represents a constant fraction of terms in the SYK Hamiltonian, and high-energy states of this model remain high-energy states of the standard SYK model.

The algorithm acts on the two-color model by first considering an additional set of $n_2$ ancilla Majorana modes with associated operators $\sigma_1, ..., \sigma_{n_2}$. The algorithm first constructs the Gaussian state $\rho_0$ optimizing the Hamiltonian $h' \triangleq \frac{i}{\sqrt{n_2}} \sum_{m=1}^{n_2} \sigma_m \chi_m$, which is the maximally entangled state on the $\sigma$ and $\chi$ modes. As this state is Gaussian, it is efficient to prepare by a classical computer. Our goal is to modify this Gaussian state into one that optimizes our original Hamiltonian using a quantum computer.

The next step of the algorithm is to apply a unitary transformation on our Gaussian state to get $\rho_\theta \triangleq e^{-\theta \zeta}\rho_0 e^{\theta \zeta}$, where $\zeta \triangleq \sum_{m=1}^{n_2} \tau_m \sigma_m$. Intuitively, we can think of this unitary transformation as a ``rotation'' of the $\sigma$ operators towards the $\tau$ operators in the Heisenberg picture, where the magnitude of the rotation is parameterized by the evolution time $\theta$. To first-order in $\theta$ this rotation achieves a high energy, and the main technical challenge is to upper-bound the effect of higher-order terms. Hastings--O'Donnell \cite{hastings2023optimizing} show that for some constant $\theta\gets\theta^\ast$ this is indeed the case for standard SYK, and $\rho_{\theta^\ast}$ achieves $\Omega(\sqrt{n})$ energy on both the two-color model and the original SYK model. 

The rotation is also crucially ``non-Gaussian'', meaning the resulting state is not a Gaussian state, and it achieves a better energy scaling than any possible Gaussian state. Hastings and O'Donnell also show that their algorithm is efficiently implementable on a quantum computer with circuit complexity polynomial in $n$.

In \cref{sec:sparse-ho-alg} we show that this algorithm retains its performance guarantees when applied to sparse SYK Hamiltonians when $p \geq \Omega(\log n/n)$. The proof of this relies on defining a sparse two-color SYK model similar to the model in \cite{hastings2023optimizing}. From here we apply similar analysis as the original Hastings--O'Donnell algorithm, with special care to show that the Bernoulli random variables do not affect the tails of our distribution significantly when $p$ is sufficiently large.

Combining \cref{thm:sparse-ho} with \cref{thm:gaussian-upper-bound}, we show a continuous regime of sparse SYK models where we have a provable separation between states with efficient classical representations and general quantum states for the optimal achievable ratio of approximation algorithms. This separation holds for all sparse SYK Hamiltonians with $p \geq C\log n/n$, extending the analogous result shown by \cite{hastings2023optimizing} for the special case of $p=1$.

\subsection{Related work}
We would like to mention the concurrent, independent work of Basso--Chen--Dalzell \cite{basso2024dissipation} that was released while we were preparing our manuscript. They consider a Hamiltonian ensemble similar to the sparse SYK model where the $J_I$ Gaussian random variables in \cref{eqn:SSYK} are replaced with Rademacher random variables, and instead of each term being sparsified independently, the ensemble has a fixed number of non-zero terms $m$ and is drawn uniformly from the set of all possible subsets of size $m$. 

They show similar circuit complexity lower bounds for states achieving a constant approximation (also based on the works of \cite{chifang2023sparserandom, anschuetz2024strongly}) and universality results showing a $\Theta(\sqrt{n})$ maximum eigenvalue for their Hamiltonian model. However, they do not consider specific bounds for Gaussian states like we do in our work. The main contribution of \cite{basso2024dissipation} is a quantum Gibbs sampling-based algorithm which produces a quantum state achieving a constant approximation to the maximum eigenvalue, for all $m>Cn\log n$. While we can only show the Hastings--O'Donnell algorithm succeeds for $m\gtrsim n^3\log n$, the algorithm and analysis is comparatively much simpler. Additionally, previous work has given analytic evidence that optimization algorithms based on Gibbs state preparation are non-optimal in preparing near-optimal states of certain disordered systems, in that there exist other algorithms capable of optimizing to even better energies~\cite{8948630,9996629}. For this reason, it is important to analyze a variety of algorithms capable of achieving a constant approximation ratio.
\section{Preliminaries}\label{sec:setup}
We use $C$ to denote positive constant(s) independent of system size $n$, unless otherwise stated. $C$ may take different values in different equations, however all such instances can be upper-bounded by some absolute constant $C^* = O(1)$. We use $\widetilde{O}$, $\widetilde{\Omega}$, $\gtrsim$, $\lesssim$, etc., to hide logarithmic factors in $n$.

We say a sequence of events $\{A_n\}_n$ occur ``with probability exponentially close to 1'' if there exists a constant $c>0$ such that $\Pr(A_n) \geq 1-\exp(-\Omega(n^c))$ and occur ``with high probability'' if there exists a constant $\alpha\geq 1$ such that $\Pr(A_n) \geq 1-1/O(n^{\alpha})$. At times, to ease notation, the sequence of events $\{A_n\}_n$ will be denoted by a single random variable $A$ and the value of $n$ is inferred from context.

\paragraph{Operators and matrices} We use the following notation for operators and matrices: $A^\top$, $A^\dagger$, and $\Tr[A]$ denote the transpose, conjugate transpose, and trace of an operator/matrix $A$ respectively. Given a $2n\times 2n$ antisymmetric matrix $A$, its Pfaffian is defined as
\begin{equation}
    \mathrm{Pf}(A) \triangleq \frac{1}{2^nn!}\sum_{\sigma \in S_{2n}} \mathrm{sgn}(\sigma)\prod_{j=1}^n A_{\sigma(2j-1), \sigma(2j)},
\end{equation}
where $S_{2n}$ denotes the symmetric group of order $2n$ and $\mathrm{sgn}(\sigma)$ denotes the parity of permutation $\sigma$. Instead of considering all permutations, one can also calculate the Pfaffian with all possible unordered pairings of $[2n]$ of the form $\pi=\{(\pi_L^1, \pi_R^1), (\pi_L^2, \pi_R^2),\dots, (\pi_L^{2k},\pi_R^{2k})\}$, where $\pi_L^1<\pi_L^2<\dots<\pi_L^{n}$ and $\pi_L^j<\pi_R^j \quad \forall j\in[n]$. Let $\mathcal{P}_{2n}$ denote the set of all such pairings. The Pfaffian can thus be written as
\begin{equation}
    \mathrm{Pf}(A) = \sum_{\pi \in \mathcal{P}_{2n}} \mathrm{sgn}(\pi) \prod_{j=1}^n  A_{\pi_L^j,\pi_R^j}
\end{equation}
where $\mathrm{sgn}(\pi)$ is the parity of the permutation $\begin{pmatrix}
1 & 2& 3 & 4 & \dots & 2n-1 & 2n \\
\pi_L^1 & \pi_R^1 & \pi_L^2 & \pi_R^2& \dots & \pi_L^n & \pi_R^n
\end{pmatrix}$.

\paragraph{Norms and metrics} We use the following notation for norms and metrics:
\begin{itemize}
    \item $\norm{A}{\mathrm{op}} \triangleq \sqrt{\lambda_{\mathrm{max}}(A^*A)}$ denotes the spectral/operator norm. Here, $\lambda_{\mathrm{max}}(M)$ denotes the largest eigenvalue of an operator or matrix $M$.
    \item $\norm{A}{F} \triangleq \sqrt{\sum_{i}\sum_j |A_{ij}|^2}=\sqrt{\Tr[AA^{\dagger}]}$ denotes the Frobenius norm.
    \item $\norm{X}{p} \triangleq \E[|X|^p]^{1/p}$ denotes the $L_p$-norm on random variables.
    \item  $|\Vec{v}|_p \triangleq \left(\sum_i|\Vec{v}_i|^p\right)^{1/p}$ denotes the $\ell_p$-norm on vectors.
    \item $T(\rho, \sigma) \triangleq \frac{1}{2} \Tr\left[\sqrt{(\rho - \sigma)^\dagger(\rho - \sigma)}\right]$ denotes the trace distance between density matrices $\rho$ and $\sigma$ of quantum states.
\end{itemize}

\subsection{Concentration bounds}\label{sec:concentration}
We will use the two following standard probability tail bounds:
\begin{itemize}
    \item Binomial Distribution Chernoff Bound \cite[Theorem 18.6]{harchol2023introduction}: If $X \sim \Binomial{n,p}$ is a Binomial random variable, for all $\delta>0$,
    \begin{equation}\label{eqn:full-binomial-chernoff}
        \Pr[X \geq (1+\delta)np] \leq \left(\frac{e^{\delta}}{(1+\delta)^{1+\delta}}\right)^{np}.
    \end{equation}

    From this we can derive the simplified bound
    \begin{equation}\label{eqn:simple-binomial-chernoff}
        \Pr[X \geq (1+\delta)np] \leq e^{-\frac{\delta^2 np}{2+\delta}}.
    \end{equation}

    \item Gaussian Tail Bound: If $X \sim \mathcal{N}(\mu,\sigma^2)$ is a Gaussian random variable, for all $t>0$,
    \begin{equation}\label{eqn:gaussian-chernoff}
        \Pr[|X -\mu|\geq t] \leq \exp\left(-\frac{t^2}{2\sigma^2}\right).
    \end{equation}
\end{itemize}

As a specific application, consider $X \sim \Binomial{O(n^a),p}$. If $n^ap > \omega(1)$, then $X \leq O(n^ap)$ with probability exponentially close to 1. If $n^ap = \Theta(1)$, then for any constant $C'>0$, there exists a constant $C$ such that
\begin{equation}
    \Pr[X \geq C\log n] \leq \left(\frac{e^{C'\log n}}{(1+C'\log n)^{(1+C'\log n)}}\right)^{\Theta(1)} \leq 1/(\log n)^{\Omega(\log n)} \leq 1/\mathrm{poly}(n),
\end{equation}
so $X \leq O(\log n)$ with high probability.

\subsection{Fermionic operators}
The creation and annihilation operators are conjugate operators that add and remove, respectively, a single particle from a given mode. Fermionic operators follow the canonical anticommutation relations. Specifically, representing the fermionic creation and annihilation operators for a mode $i$ as $a^\dag_i$ and $a_i$, respectively, we have:
\begin{equation}
    \{a_i, a_j\}=0, \qquad \{a^\dag_i, a^\dag_j\}=0, \qquad \{a_i, a_j^\dagger\}=\delta_{ij},
\end{equation}
where $\{\cdot,\cdot\}$ denotes the anticommutator.

Majorana fermions are a special class of fermions which are their own antiparticle. Mathematically, they have equivalent creation and annihilation operators. That is, given a Majorana annihilation operator $\gamma_i$ we have $\gamma_i = \gamma^\dag_i$. Majorana operators are also traceless, square to the identity, and form a Clifford algebra, i.e., $\{\gamma_i, \gamma_j\} = 2\delta_{ij}\mathbb{I}$. We can represent a general set of $n$ fermionic modes with $2n$ Majorana operators as follows:
\begin{equation}
    a_i = \frac{1}{2}(\gamma_{2i-1} + i\gamma_{2i}), \qquad a^\dag_i = \frac{1}{2}(\gamma_{2i-1} - i\gamma_{2i}) \qquad \forall i = 1,2,\ldots,n,
\end{equation}
and for this reason we will concern ourselves only with Majorana fermions in what follows.

\paragraph{Canonical representation.}
There exists a canonical matrix representation of the Majorana operators $\gamma_i$. The Weyl--Brauer matrices form a representation of the Clifford algebra, and generalize the Pauli matrices to higher dimensions \cite{weyl-brauer-1935}. 
\begin{definition}[Weyl--Brauer matrices]
    The degree-$2n$ Weyl--Brauer matrices are matrices in $\mathbb{C}^{2^n \times 2^n}$  defined as:
    \begin{align*}
        \gamma_{1} &= X_1 \otimes \mathbb{I}_2 \otimes \mathbb{I}_3 \otimes \ldots \otimes \mathbb{I}_n,\\
        \gamma_{2} &= Y_1 \otimes \mathbb{I}_2 \otimes \mathbb{I}_3 \otimes \ldots \otimes \mathbb{I}_n,\\
        \gamma_{3} &= Z_1 \otimes X_2 \otimes \mathbb{I}_3 \otimes \ldots \otimes \mathbb{I}_n,\\
        \gamma_{4} &= Z_1 \otimes Y_2 \otimes \mathbb{I}_3 \otimes \ldots \otimes \mathbb{I}_n,\\
        &\vdots
        \\
        \gamma_{2n-1} &= Z_1 \otimes Z_2\otimes Z_3 \otimes \ldots \otimes X_n,\\
        \gamma_{2n} &= Z_1 \otimes Z_2 \otimes Z_3 \otimes \ldots \otimes Y_n,
    \end{align*}
    where $X_i,Y_i,Z_i$ are $2\times 2$ Pauli matrices on the $i$th qubit and $\mathbb{I}_i$ is the $2\times 2$ identity matrix on the $i$th qubit.
\end{definition}

At times, given a subset $I \subseteq [2n]$ of size $|I|=q$ with elements $I=\{i_1, \dots, i_q\}$ ordered such that $i_1 < i_2 < \cdots <i_q$, we will use the subscript $\gamma_I$ to denote Hermitian monomial of Majorana operators:
\begin{equation}
    \gamma_I \triangleq i^{q/2}\gamma_{i_1} \cdots \gamma_{i_q}.
\end{equation}

\subsection{Fermionic Hamiltonians}\label{sec:ferm_ham}

Our problem of interest is in finding the maximal energy of a system of $n$ fermionic modes (e.g., electron orbitals) with low-degree interactions.
Typically, in physics, we aim to find the \emph{minimal} energy, but due to the symmetry in the problem, we will focus on the equivalent problem of finding the maximum energy.
\begin{problem}[Degree-$q$ fermionic Hamiltonian optimization]\label{degreeqfermion}
    Given an even integer $q$ and the fermionic Hamiltonian 
    \begin{equation}
    H = \sum_{\substack{I \subseteq [2n]\\|I|=q}} c_{I}\gamma_I, \qquad c_{I} \in \mathbb{R},
    \end{equation}
    find the maximum eigenvalue $\lambda_{\max}(H) = \max_{\psi} \braket{\psi|H|\psi}$.
\end{problem}

We will particularly be interested in the case where $H$ is an instance of the \emph{sparse SYK-q (SSYK-q) model}:
\begin{equation} \label{eq:SSYK_q_equation}
    H_{\text{SSYK-q}}\left(p\right)\triangleq\binom{2n}{q}^{-1/2}p^{-1/2}\sum_{\substack{I \subseteq [2n] \\ |I|=q}} J_{I}X_{I}\gamma_I,
\end{equation}
where $q\geq 4$ is even, the $J_{I}$ are i.i.d.\ standard Gaussian random variables, and the $X_{I}$ are i.i.d.\ Bernoulli random variables with $\mathbb{P}\left(X_{I}=1\right)=p$. When $p=1$ we call this model the \emph{(standard/dense) SYK-q model}. For the special case $q=4$---the minimal setting beyond free fermions---we will refer to the dense and sparse models as the SYK and SSYK models, respectively.

Note that the ``interesting'' range of sparsification lies between $p = \Theta(1/n^{q-1})$ and $p=1$, as any value $p =o(1/n^{q-1})$ will generally have fermionic modes with no interactions with any other modes. Thus for notational convenience, we assume some minimum value $p_{min} = \epsilon/n^{q-1}$ for any arbitrarily small constant $\epsilon > 0$. We implicitly assume that $p \in [p_{min}, 1]$, and whenever we refer to ``all values of $p$'' we mean all values within this range.

\subsection{Gaussian states}
We can also consider restricting $\ket{\psi}$ in the maximization of $\braket{\psi|H|\psi}$ to be over specific classes of states. One important such class is the set of \emph{Gaussian states}, which are the ground and thermal states of quadratic fermionic Hamiltonians. Additionally, they are the outputs of mean-field methods (e.g., Slater determinants which are outputs of Hartree-Fock are a subset of Gaussian states) for estimating the ground state energies of fermionic systems and are efficiently-representable classically (for a survey on fermionic Gaussian states, see \cite{Surace_2022}).

\begin{definition}[Fermionic Gaussian states]
    For any orthogonal matrix $R \in SO(2n)$, we define the unitary operator $U_R\in U(2^n)$ acting on $n$ fermionic modes via its action on the Weyl--Brauer representation as:
    \begin{equation}\label{eqn:gaussian-def}
        (U_R)^\dag \gamma_i (U_R) = \sum_{j=1}^{2n} R_{ij}\gamma_j \qquad \forall i = 1,2,\ldots,2n,
    \end{equation}
    which is uniquely defined up to global phase \cite{Bravyi_2019}. A state $\ket{\psi}$ is a (pure) fermionic Gaussian state if it has the form $\ket{\psi} = U_R\ket{0}^{\otimes n}$.
\end{definition}
Any fermionic Gaussian state can also be written in the standard density matrix form:
\begin{equation}\label{eqn:density-matrix}
    \rho = \frac{1}{2^n} \prod_{j=1}^n \left(\mathbb{I} + i\lambda_j\Tilde{\gamma}_{2j-1}\Tilde{\gamma}_{2j}\right),
\end{equation}
where $\lambda_j \in [-1,1]$ and $\Tilde{\gamma}_i = (U_R)^\dag \gamma_i (U_R)$.  If $\lambda_j \in \{-1,1\}$ the state is a pure fermionic Gaussian state, otherwise it is a mixed fermionic Gaussian state. We generally use $\ket{\psi}$ to denote pure quantum state vectors and $\rho$ to denote density matrices of mixed quantum states.
The \emph{covariance matrix} $\Gamma$ of a Gaussian state $\rho$ collects the two-body expectations of $\rho$ in a matrix as defined below.
\begin{definition}[Gaussian state covariance matrix]
    Given a fermionic Gaussian state $\rho$, we define its \emph{covariance matrix} $\Gamma \in \mathbb{R}^{2n \times 2n}$ as $\Gamma_{ij} \triangleq \frac{i}{2}\Tr(\rho[\gamma_i, \gamma_j])$. $\Gamma$ can be block-diagonalized as 
    \begin{equation}
        \Gamma = R\, \bigoplus_{j=1}^n \begin{pmatrix}
                                        0 & \lambda_j\\
                                        -\lambda_j & 0
                                        \end{pmatrix} R^\top         
    \end{equation}
    where $R \in SO(2n)$ is the orthogonal matrix from \cref{eqn:gaussian-def} and $\lambda_j \in [-1,1]$ is from \cref{eqn:density-matrix}. We have $\Gamma^\top\Gamma = \mathbb{I}$ for all pure fermionic Gaussian states.
\end{definition}

Gaussian states are classically useful since their energy on products of Majorana operators can be efficiently computed classically using Wick's Theorem \cite{Wick_1950}:
\begin{proposition}[Wick's theorem for fermionic Gaussian states]\label{prop:wicks-theorem}
    Given a Gaussian state $\rho_\Gamma$ with covariance matrix $\Gamma$, for any $I\subseteq [2n]$ where $|I|$ is even,
    \begin{equation}
        \Tr(\gamma_I\rho_\Gamma) = \mathrm{Pf}(\Gamma_I)
    \end{equation}
    where $\Gamma_I$ is the principal submatrix of $\Gamma$ obtained by only keeping the rows and columns indexed by the indices in $I$.
\end{proposition}

As classical mean-field methods output Gaussian states, an interesting open question is how well Gaussian states can approximate the maximum energy over arbitrary states for degree-4 Hamiltonians. 

\begin{definition}[Gaussian approximation factor]
    Denote 
    \begin{equation}
        \lambda_{\mathrm{Gauss}}(H) \triangleq \max_{\text{Gaussian states } \psi} \braket{\psi|H|\psi},
    \end{equation}
    i.e., the maximum energy where $\ket{\psi}$ is restricted to be a Gaussian state. The \emph{Gaussian approximation ratio} $\alpha_{\mathrm{Gauss}}\in [0,1]$ of Gaussian states on a fermionic Hamiltonian $H$ is defined as $\alpha_{\mathrm{Gauss}} \triangleq \lambda_{\mathrm{Gauss}}(H)/\lambda_{\max}(H)$.
\end{definition}

When considering the sparse SYK Hamiltonian, we will often be interested in a net of Gaussian states in the following sense. The relation between the gate count of a parameterized quantum circuit class and a net over the class appears in the proof of \cite[Lemma~D.1]{chifang2023sparserandom}, which we reproduce here for completeness. However this general argument is folklore, e.g., the Solovay--Kitaev theorem.
\begin{proposition}[Gaussian state net]\label{prop:gaussian-net}
    For any constant $0 < \epsilon < 1/2$, there exists a fixed set $S$ of $\exp(O(n^2\log (n/\epsilon)))$ pure Gaussian states which is an $\epsilon$-net over all $n$-mode fermionic pure Gaussian states in the following sense. For any pure Gaussian state $\ket{\psi}$,
    \begin{equation}
        \inf_{\ket{\phi}\in S} T(\ket{\phi}\bra{\phi}, \ket{\psi}\bra{\psi}) \leq \epsilon.
    \end{equation}
\end{proposition}
\begin{proof}
    The set of all $n$-mode fermionic Gaussian states is generated by a circuit architecture of at most $O(n^2)$ continuously parameterized two-qubit gates \cite{Jiang_2018}. For any Gaussian state $\ket{\phi}$, we can explicitly associate a set of $O(n^2)$ parameters $\{\theta_i\}$ which express the state as 
    \begin{equation}
        \ket{\phi} = \prod_{j=1}^{O(n^2)} U_j(\theta_j) \ket{0},
    \end{equation}
    from some fixed starting state $\ket{0}$.
    We form our net by discretely parameterizing each gate parameter $\theta_i$ with intervals of size $\epsilon' = \Theta(\epsilon/n^2)$ so that there are $O(n^2/\epsilon)$ discrete values each parameter can take. Given a state $\ket{\phi}$ with gate parameters $\{\theta_i\}$, denote as $\bar{\theta}_i$ the closest parameter value in our net to $\theta_i$. We now show that $\ket{\bar{\phi}}$ which is the state associated to parameters $\{\bar{\theta}_i\}$ is $O(\epsilon)$ away in trace distance from $\ket{\phi}$. Denote intermediate states for $0\leq k \leq | \{\theta_i\} |$ as
    \begin{equation}
        \ket{\psi_k} = \prod_{j=1}^k U_j(\theta_j) \prod_{j'=k+1}^{|\{\theta_i\}|} U_{j'}(\bar{\theta}_{j'}) \ket{0}.
    \end{equation}
    Then, it holds that
    \begin{equation}
        T(\ket{\phi}\bra{\phi}, \ket{\bar{\phi}}\bra{\bar{\phi}} ) 
        \leq 
        \sum_{j=0}^{| \{\theta_i\} |-1} T(\ket{\psi_j}\bra{\psi_j}, \ket{\psi_{j+1}}\bra{\psi_{j+1}} ) \leq O(\epsilon' n^2) = O(\epsilon).
    \end{equation}
    In total, this construction will have $\exp(O(n^2 \log(n/\epsilon)))$ many distinct states in the net. 
\end{proof}

Given any Hamiltonian $H$, we have that $T(\ket{\phi}\bra{\phi}, \ket{\psi}\bra{\psi}) \leq \epsilon/m$ implies
\begin{equation}
    \big\lvert\braket{\phi|H|\phi} - \braket{\psi|H|\psi}\big\rvert \leq \frac{\epsilon}{m} \cdot \norm{H}{\mathrm{op}}.
\end{equation}
Thus if $\norm{H}{\mathrm{op}} \leq m$, our $(\epsilon/m)$-net construction approximates all Gaussian states in energy with respect to $H$ up to an additive $\epsilon$ error.

\subsection{Commutation index and Lovász theta function}
Given a random Hamiltonian of the form $H=\sum_I J_I A_I$, where $A_I$ are Hermitian operators which square to the identity operator and $J_I$ are random coefficients with mean zero and unit variance, a key quantity of study is the variance of the energy of a given state $\ket{\psi}$:
\begin{equation}
    \operatorname{Var}(\bra{\psi}H\ket{\psi}) = \sum_I \bra{\psi}A_I\ket{\psi}^2.
\end{equation}
We will later use upper bounds on the energy variance to obtain tail bounds on the probability that the energy of a state can be at a specified value. We describe here how to bound the variance quantity above using the Lovász theta function of the so-called \emph{commutation graph} of the set $\mathcal{S}=\{A_I\}_I$. This is the graph constructed by taking $\mathcal{S}$ as the vertex set, and inserting an edge between $A_I,A_J\in\mathcal{S}$ if and only if they anticommute.
\begin{definition}[Commutation graph] \label{def:commutation_graph}
Given a finite set $\mathcal{S}$ of Pauli or Majorana operators, the commutation graph $G(\mathcal{S})$ of $\mathcal{S}$ is an undirected graph of $|\mathcal{S}|$ nodes constructed as follows:
\begin{itemize}
    \item The vertices of $G(\mathcal{S})$ correspond to operators $A_i \in \mathcal{S}$.
    \item There is an edge between two vertices if and only if their operators \emph{anticommute} (i.e., $i \sim j$ if $A_iA_j + A_jA_i = 0$).
\end{itemize}
\end{definition}
We follow closely the formulation of \cite{anschuetz2024strongly}, though similar ideas have been advanced in works on studying uncertainty relations \cite{de2023uncertainty,xu2024bounding}, random Hamiltonians \cite{anschuetz2024bounds}, and shadow tomography \cite{king2024triply}.
The commutation index is defined as follows:
\begin{definition}[Commutation index]
The commutation index $\Delta$ of a finite set of Hermitian operators $\mathcal{S}$ is
\begin{equation}
\Delta(\mathcal{S}) \triangleq \frac{1}{|\mathcal{S}|} \sup_{|\psi\rangle} \; \sum_{A_i \in \mathcal{S}} \langle\psi|A_i|\psi\rangle^2.
\end{equation}
\end{definition}
It is easy to see that the commutation index bounds the energy variance of any given state. In turn, the commutation index is known to be bounded by the \emph{Lovász theta function} of the commutation graph of $\mathcal{S}$. This fact was proven in \cite[Equation 5]{de2023uncertainty}, \cite{hastings2023optimizing}, and \cite[Proposition 3]{xu2024bounding}. 
We define the Lovász theta function as follows, noting that other equivalent formulations exist as well.
\begin{definition}[Lovász theta function] \label{def:lovasz}
Let $G=(V,E)$ be an undirected graph on $|V|$ vertices with edges in $E$. The \emph{Lovász theta function} $\vartheta(G)$ is given by a semidefinite program of dimension $|V|$:  
\begin{equation}
\begin{split}
    \max \ \{ \ &\Tr\left(\mathbb{J} X\right) \ , \ X \in \mathbb{R}^{|V| \times |V|} \\
    &\text{s.t.} \ X \succeq 0 \ , \ \Tr\left(X\right) = 1 \ , \ X_{jl} = 0 \ \forall (j,l) \in E \ \}, \label{eq:theta_primal}    
\end{split}
\end{equation}
where $\mathbb{J}$ is the all-ones matrix.
\end{definition}

We note that for any induced subgraph $H\subset G$, we have $\vartheta(H)\leq \vartheta(G)$, i.e., removing vertices from a graph cannot increase its Lovász theta function. This gives a (not necessarily tight) upper bound on the Lovász theta function when sparsifying an arbitrary graph by deleting its vertices.

Finally, we restate a key lemma which shows that $\Delta(\mathcal{S}) \leq \vartheta(G(\mathcal{S}))/|\mathcal{S}|$.

\begin{lemma}[Bound on commutation index \cite{de2023uncertainty,hastings2023optimizing,xu2024bounding}]\label{lem:bound_via_lovasz}
    Given a finite set $\mathcal{S}$ of Pauli or Majorana operators with commutation graph $G(\mathcal{S})$, it holds that
    \begin{equation}
        \Delta(\mathcal{S}) \leq \frac{\vartheta(G(\mathcal{S}))}{|\mathcal{S}|}.
    \end{equation}
\end{lemma}
This lemma allows us to bound the commutation index of the SYK Hamiltonian (and consequently the SSYK Hamiltonian). 
\begin{lemma}[Bound on commutation index of SYK Hamiltonian \cite{anschuetz2024strongly}]\label{lem:syk_bound_via_lovasz}
    Let $\mathcal{S}$ be the set of all degree $q$ Majorana operators for $q>0$ even:
    \begin{equation}
        \mathcal{S}= \left\{ i^{q/2} \gamma_{i_1}\gamma_{i_2}\ldots\gamma_{i_q} \;|\;  \{i_1,i_2,\ldots,i_q\} \subseteq [2n] , \; i_1 < i_2 < \cdots <i_q  \right\}.
    \end{equation}
    Then $\Delta(\mathcal{S}) = \Theta(n^{-q/2})$.
\end{lemma}
\section{Maximum eigenvalue of sparse SYK Hamiltonians}
\label{sec:universality}
We first prove the following ``universality'' theorem, which states that the maximum energy eigenstate for the sparse SYK model has the same energy as the maximum energy eigenstate for the SYK model up to constant factors with probability exponentially close to 1. This result holds independently of the sparsification parameter.\footnote{Up to taking $p\geq p_{min}$ as described in \Cref{sec:ferm_ham}.}

\begin{theorem}\label{thm:ssyk-energy-concentration}
    Given the SSYK model, for all sparsification parameters $p=1/n^{a}$ when $a<3$, and any $\epsilon > 0$, there exists some constant $c>0$ such that
    \begin{equation}
        \Pr(|\lambda_{\mathrm{max}}(H_{SSYK}(p)) - \E \lambda_{\mathrm{max}}(H_{SYK})| > \epsilon \sqrt{n}) < \exp(-\Omega(n^c)).
    \end{equation}
\end{theorem}

To show \cref{thm:ssyk-energy-concentration}, we first show an intermediate lemma which states that the expected maximum eigenvalue of SSYK for any $p>\omega(1/n^3)$ matches the expected maximum eigenvalue of SYK asymptotically.
\begin{lemma}\label{lemma:ssyk-expected-max}
    Given the SSYK model, for all sparsification parameters $p=1/n^{a}$ when $a<3$, and any $\epsilon > 0$,
    \begin{equation}
        \left| \E\left[\lambda_{\mathrm{max}}(H_{SSYK})\right] - \E\left[\lambda_{\mathrm{max}}(H_{SYK})\right] \right| = o(\sqrt{n}).
    \end{equation}
    where the expectation is over the disorder of the Hamiltonians.
\end{lemma}
\begin{proof}
SYK has an expected maximum eigenvalue of $\Theta(\sqrt{n})$ \cite{hastings2023optimizing}. From Theorem 6 of \cite{anschuetz2024bounds}, the expected maximum eigenvalue of $H_{SSYK}$ matches the standard SYK expected maximum eigenvalue when $n\rightarrow \infty$ under the condition that the disorder $\alpha_I = p^{-1/2}X_IJ_I$ obeys the following conditions:
\begin{enumerate}
    \item First moment: $\E[\alpha_I] = 0$;
    \item Second moment: $\E[\alpha_I^2] = 1$;
    \item Boundedness: For all $\epsilon > 0$: $\lim_{n \rightarrow \infty} \frac{1}{m_n} \E\left[\sum_{I} \alpha_I^2;\; |\alpha_I| > \epsilon\sqrt{\frac{m_n}{n}}\right] = 0$,
\end{enumerate}
where $m_n = \binom{2n}{4}$. Since $\alpha_I^2 = p^{-1}X_I J_I^2$, the boundedness condition is satisfied when $p > \omega(1/n^{3})$. 
\end{proof}

We next show that this result in expectation extends to a result with probability exponentially close to 1 using concentration bounds. Consider the sparse SYK model which as a reminder takes the form:
\begin{equation}
    H_{SSYK} \triangleq \binom{2n}{4}^{-1/2}\sum_{\substack{I\subset [2n]\\|I|=4}} (p^{-1/2}X_IJ_I) \gamma_I.
\end{equation}
To prove that the maximum eigenvalue of SSYK concentrates with probability exponentially close to 1 around its expectation, we apply Gaussian concentration for Lipschitz functions.
\begin{lemma}[Gaussian concentration inequality for Lipschitz functions; Theorem 2.26 of \cite{wainwright2019high}]
    Let $J_1,...,J_m$ be i.i.d. standard Gaussian variables and $f:\R^m\rightarrow \R$ be an $L$-Lipschitz function. Then for any $t \geq 0$:
    \begin{equation}
        \Pr[|f(X) - \E f(X)| \geq t] \leq 2\exp\left(-\frac{t^2}{2L^2}\right).
    \end{equation}
\end{lemma}
Fix some value $m$ denoting the number of non-zero entries of the binomial vector. Let $f$ denote the function mapping $\Vec{J} \mapsto \binom{2n}{4}^{-1/2}p^{-1/2} \cdot \lambda_{\mathrm{max}}\left(\sum_{i=1}^m  J_i \gamma_i\right)$, and let $\ket{\psi}$ be any state. Then
\begin{equation}
    \left|\nabla_{\Vec{J}}\binom{2n}{4}^{-1/2}p^{-1/2}\bra{\psi}\left(\sum_{i=1}^m  J_i \gamma_i\right)\ket{\psi}\right|_2^2 \leq \binom{2n}{4}^{-1}p^{-1}\sum_{i=1}^m \bra{\psi}\gamma_i\ket{\psi}^2 \leq O(n^{-2}p^{-1}),
\end{equation}
where the last inequality is from \cref{lem:syk_bound_via_lovasz} and holds for all $m$.
As this Lipschitz bound holds for the energy over any state, it also holds for the supremum over all states which is the Lipschitz constant $L$ of $\lambda_{\mathrm{max}}(H_{SSYK}(p))$ is $O(n^{-1}p^{-1/2})$ (see, e.g., Proposition 8 of \cite{anschuetz2024bounds}). This gives:
\begin{equation}\label{eqn:ssyk-energy-concentration}
     \Pr\left(|\lambda_{\mathrm{max}}(H_{SSYK}(p)) - \E \lambda_{\mathrm{max}}(H_{SSYK}(p))\right| \geq t\sqrt{n}) \leq 2\exp\left(-\Omega\left(\frac{t^2n^3p}{2}\right)\right).
\end{equation}
Thus we see that the maximum eigenvalue lies within $t\sqrt{n}$ of its expectation with probability exponentially close to 1 for any $t>0$ when $a < 3$.
\begin{remark}
    In \cref{appendix:upper-bound}, we show an improvement can be made to the probability in which the upper bound holds when $2 < a < 3$ compared to  \cref{eqn:ssyk-energy-concentration}. 
\end{remark}

We now complete the proof of \cref{thm:ssyk-energy-concentration}.
\begin{proof}[Proof of \cref{thm:ssyk-energy-concentration}]
    Fix any constant $t > 0$. From \cref{lemma:ssyk-expected-max}, for any constant $t'$ such that $0<t'<t$, there exists an $n_0$ such that for all $n \geq n_0$, 
    \begin{equation}
        \left| \E\left[\lambda_{\mathrm{max}}(H_{SYK})\right] - \E\left[\lambda_{\mathrm{max}}(H_{SSYK})\right] \right| \leq t'\sqrt{n}.
    \end{equation}
    From \cref{eqn:ssyk-energy-concentration},
    \begin{equation}
        \Pr\left[|\lambda_{\mathrm{max}}(H_{SSYK}) - \E \lambda_{\mathrm{max}}(H_{SYK})| \geq (t-t')\sqrt{n}\right]\leq 2\exp\left(-\Omega\left(\frac{(t-t')^2n^3p}{2}\right)\right),
    \end{equation}
    and by triangle inequality
    \begin{align*}
        |\lambda_{\mathrm{max}}(H_{SSYK}) - \E \lambda_{\mathrm{max}}(H_{SYK})| &\leq |\lambda_{\mathrm{max}}(H_{SSYK}) - \E \lambda_{\mathrm{max}}(H_{SSYK})| \\&\qquad + |\E \lambda_{\mathrm{max}}(H_{SSYK}) - \E \lambda_{\mathrm{max}}(H_{SYK})|\\&\leq |\lambda_{\mathrm{max}}(H_{SSYK}) - \E \lambda_{\mathrm{max}}(H_{SSYK})| +t'\sqrt{n},
    \end{align*}
    again for $n\geq n_0$. Finally, we have $\Pr\left[|\lambda_{\mathrm{max}}(H_{SSYK}) - \E \lambda_{\mathrm{max}}(H_{SYK})\right| \geq t\sqrt{n})] \leq \Pr\left[\left|\lambda_{\mathrm{max}}(H_{SSYK}) - \E \lambda_{\mathrm{max}}(H_{SSYK})\right| \geq (t-t')\sqrt{n}\right]$ and 
    \begin{align*}
        \Pr\left[\left|\lambda_{\mathrm{max}}(H_{SSYK}) - \E \lambda_{\mathrm{max}}(H_{SSYK})\right| \geq (t-t')\sqrt{n}\right]&\leq 2\exp\left(-\Omega\left(\frac{(t-t')^2n^3p}{2}\right)\right)\\& \leq \exp\left(-\Omega\left({n^{3-a}}\right)\right).
    \end{align*}
    Thus for large $n$, $|\lambda_{\mathrm{max}}(H_{SSYK}) - \E \lambda_{\mathrm{max}}(H_{SYK})| < t\sqrt{n}$ occurs with probability exponentially close to 1 since $a<3$.
\end{proof}

\begin{remark}
    In the case of $a=3$, i.e. $p=\Theta(1/n^3)$, we note how the existing algorithms proposed by \cite{Herasymenko_2023, Hothem2023} outputting Gaussian states for sparse SYK when $p = \Theta(1/n^3)$ achieve $\Theta(\sqrt{n})$. This shows that with probability exponentially close to 1, the maximum eigenvalue of the sparse SYK model is $\Theta(\sqrt{n})$ when $p = \Theta(1/n^3)$, with the Gaussian state output of the above algorithms as a witness. However the expected maximum eigenvalue may differ from the dense SYK model by a constant factor, as opposed to the scaling in \cref{lemma:ssyk-expected-max}.
\end{remark}
\section{Optimal energy classical ansatz for sparse SYK Hamiltonians}\label{sec:gaussian-upper}
In this section we prove tight upper and lower bounds on the maximum energy Gaussian states can achieve on sparse SYK Hamiltonians when $p\geq \Omega(\log n/n^2)$, by showing that the Gaussian approximation ratio $\alpha_{\mathrm{Gauss}}=\Theta(1/\sqrt{n})$ for all $p\geq \Omega(\log n/n^2)$ and $\alpha_{\mathrm{Gauss}}=O\left(n^{-\frac{1+a}{2}}\log n\right)$ when $p=\Theta(1/n^{2+a})$ for any $0 \leq a \leq 1$. We also give upper bounds on the energy achievable by states with low circuit complexity and conjecture a stronger upper bound with evidence from numerical simulations.

\tikzset{
  tensor/.style = {draw, thick, minimum width=1.5cm, minimum height=1.5cm, align=center},
  leg/.style    = {thick},
  wire/.style   = {thick},
 Gnode/.style = {draw, rectangle, minimum width=0.6cm, minimum height=0.6cm, inner sep=1.6pt, thick, fill=white, font=\footnotesize},
  lbl/.style    = {fill=white, inner sep=1pt, font=\scriptsize},
   group/.style     = {draw, rounded corners=5pt, dashed, inner sep=5pt, thick},
}

\subsection{Gaussian state energy upper bound}\label{subsec:gaussian-energy-upper}
We restate \cref{thm:gaussian-upper-bound} for convenience:
\gaussianthm*

We begin by giving a detailed exposition of the results on SYK given in \cite[Section VI]{hastings2023fieldtheory}. We then extend the method to achieve results for sparse SYK. 

\paragraph{Hastings Construction}
Hastings shows that the energy expectation value for any fixed Gaussian state on $H^k$, where $H$ is a Hamiltonian drawn from SYK model, is bounded by $O(k)^{2k}$ with high probability. To do this, we fix an arbitrary Gaussian state $\rho$ and its covariance matrix $\Gamma$. Its energy on $H^k$ can be written as
\begin{equation}
    \Tr(H^k\rho) = \bigsum_{\substack{S_1, S_2,\dots, S_k\\\forall \ell: |S_\ell|=4}} \left(\prod_{j=1}^k  J_{S_j}\right) \Tr\left(\rho\prod_{j=1}^k \gamma_{S_{j}}\right)
\end{equation}
Let $S_{1:k}=(S_1 , S_2,\dots, S_k)$ be an ordered sequence of $4k$ integers (potentially with repeated elements), and let $\mathcal{P}_{4k}$ denote the set of all partitions of the set $[4k]$ into unordered pairs (i.e. pairings $\pi=\{(\pi_L^1, \pi_R^1), (\pi_L^2, \pi_R^2),\dots, (\pi_L^{2k},\pi_R^{2k})\}$ such that $\pi_L^i < \pi_R^i$ and $\pi_L^1 < \pi_L^2 < \cdots < \pi_L^{2k}$). By Wick's theorem (\cref{prop:wicks-theorem}),
\begin{align*}
    \Tr(H^k\rho) &= \bigsum_{\substack{S_1, S_2,\dots, S_k\\\forall \ell: |S_\ell|=4}} \left(\prod_{j=1}^k J_{S_j}\right) \left(\sum_{\pi\in \mathcal{P}_{4k}}\left(\mathrm{sgn}(\pi)\prod_{j=1}^{2k}\Gamma_{S_{1:k}[\pi_{L}^j], S_{1:k}[\pi_{R}^j]} \right)\right) \\&= \bigsum_{\pi\in \mathcal{P}_{4k}}\mathrm{sgn}(\pi)\sum_{\substack{S_1, S_2,\dots, S_k\\\forall \ell: |S_\ell|=4}} \left(\prod_{j=1}^{k} J_{S_j}\Gamma_{S_{1:k}[\pi_{L}^{2j-1}], S_{1:k}[\pi_{R}^{2j-1}]}\Gamma_{S_{1:k}[\pi_{L}^{2j}], S_{1:k}[\pi_{R}^{2j}]} \right).
\end{align*}
Additionally, the SYK disorder entries $J_I$ can be represented by an associated tensor $\mathbf{J}$, as follows:
\begin{definition}[SYK tensor \cite{hastings2023optimizing}]
    Given an SYK-4 Hamiltonian $H_{SYK}=\sum_{\substack{I\subseteq [2n]\\|I|=4}}J_I\gamma_I$, define its associated 4-index tensor $\mathbf{J}$ such that given $I=\{i_1, i_2, i_3, i_4\}$, $i_1 < i_2 < i_3 < i_4$, and any permutation $\sigma \in S_4$, $\mathbf{J}_{\sigma(i_1),\sigma(i_2),\sigma(i_3),\sigma(i_4)} \triangleq \mathrm{sgn}(\sigma)J_I$. $\mathbf{J}$ is therefore a completely antisymmetric 4-index tensor by construction.
\end{definition}

The key insight of Hastings is that each pairing term in the summation over $\mathcal{P}_{4k}$ can be represented as a single tensor network. Fix a pairing $\pi=\{(\pi_L^1, \pi_R^1), (\pi_L^2, \pi_R^2),\dots, (\pi_L^{2k},\pi_R^{2k})\}$. We first add $k$ copies of the $\mathbf{J}$ tensor, each with 4 legs; label the set of all $4k$ legs with indices $1,2\dots,4k$. For each pair $(\pi_L^j, \pi_R^j)$ in the pairing $\pi$, we connect the legs $J_{\pi_L^j}$ and $J_{\pi_R^j}$ with the 2-leg tensor $\Gamma$ representing the covariance matrix of the Gaussian state. This tensor network contraction has the value 
\begin{equation}
    \mathrm{sgn}(\pi)\sum_{\substack{S_1, S_2,\dots, S_k\\\forall \ell: |S_\ell|=4}} \left(\prod_{j=1}^{k} J_{S_j}\Gamma_{S_{1:k}[\pi_{L}^{2j-1}], S_{1:k}[\pi_{R}^{2j-1}]}\Gamma_{S_{1:k}[\pi_{L}^{2j}], S_{1:k}[\pi_{R}^{2j}]} \right)
\end{equation}
exactly (for a simple example, see \cref{fig:SYK-tensor-network}). 

\begin{figure}[ht]
     \centering
     \begin{subfigure}[c]{0.3\textwidth}
         \centering
        \begin{tikzpicture}
          \node[tensor] (J) at (0,0) {$\mathbf{J}$};
          \coordinate (i1) at ($(J.west)+(-0.6,  0.6)$);
          \coordinate (i2) at ($(J.east)+( 0.6,  0.6)$);
          \coordinate (i3) at ($(J.west)+(-0.6, -0.6)$);
          \coordinate (i4) at ($(J.east)+( 0.6, -0.6)$);
        
          \draw[leg] ($(J.west)+(0, 0.6)$) -- (i1) node[left] {$1$};
          \draw[leg] ($(J.east)+(0, 0.6)$) -- (i2) node[right] {$2$};
          \draw[leg] ($(J.west)+(0,-0.6)$) -- (i3) node[left] {$3$};
          \draw[leg] ($(J.east)+(0,-0.6)$) -- (i4) node[right] {$4$};
        
          \node[Gnode] (G12) at ($(i1)!.5!(i2)+(0, 1)$) {$\Gamma$};
          \node[Gnode] (G34) at ($(i3)!.5!(i4)+(0, -1)$) {$\Gamma$};
        
          \draw[wire] (i1) to[bend left=30] (G12);
          \draw[wire] (i2) to[bend right=30] (G12);
          \draw[wire] (i3) to[bend right=30] (G34);
          \draw[wire] (i4) to[bend left=30] (G34);

        \end{tikzpicture}
         \caption{$\{(1,2), (3,4)\}$}
         \label{fig:12-34}
     \end{subfigure}
     \hfill
     \begin{subfigure}[c]{0.3\textwidth}
         \centering
         \begin{tikzpicture}
          \node[tensor] (J) at (0,0) {$\mathbf{J}$};
          \coordinate (i1) at ($(J.west)+(-0.6,  0.6)$);
          \coordinate (i2) at ($(J.east)+( 0.6,  0.6)$);
          \coordinate (i3) at ($(J.west)+(-0.6, -0.6)$);
          \coordinate (i4) at ($(J.east)+( 0.6, -0.6)$);
        
          \draw[leg] ($(J.west)+(0, 0.6)$) -- (i1) node[left] {$1$};
          \draw[leg] ($(J.east)+(0, 0.6)$) -- (i2) node[right] {$2$};
          \draw[leg] ($(J.west)+(0,-0.6)$) -- (i3) node[left] {$3$};
          \draw[leg] ($(J.east)+(0,-0.6)$) -- (i4) node[right] {$4$};
        
          \node[Gnode] (G13) at ($(J.west)+(-0.5,0)$) {$\Gamma$};
          \node[Gnode] (G24) at ($(J.east)+(0.5,0)$) {$\Gamma$};
        
          \draw[wire] (i1) to[bend right=15] (G13);
          \draw[wire] (i3) to[bend left=15] (G13);
          \draw[wire] (i2) to[bend left=15] (G24);
          \draw[wire] (i4) to[bend right=15]  (G24);
        
        \end{tikzpicture}
         \caption{$\{(1,3), (2,4)\}$}
         \label{fig:13-24}
     \end{subfigure}
     \hfill
     \begin{subfigure}[c]{0.3\textwidth}
         \centering
         \begin{tikzpicture}
          \node[tensor] (J) at (0,0) {$\mathbf{J}$};
          \coordinate (i1) at ($(J.west)+(-0.6,  0.6)$);
          \coordinate (i2) at ($(J.east)+( 0.2,  0.6)$);
          \coordinate (i3) at ($(J.west)+(-0.6, -0.6)$);
          \coordinate (i4) at ($(J.east)+( 0.6, -0.6)$);
        
          \draw[leg] ($(J.west)+(0, 0.6)$) -- (i1) node[left] {$1$};
          \draw[leg] ($(J.east)+(0, 0.6)$) -- (i2) node[right] {$2$};
          \draw[leg] ($(J.west)+(0,-0.6)$) -- (i3) node[left] {$3$};
          \draw[leg] ($(J.east)+(0,-0.6)$) -- (i4) node[right] {$4$};
        
          \node[Gnode] (G14) at ($(i1)!.5!(i2)+(0, 1)$) {$\Gamma$};
          \node[Gnode] (G23) at ($(i3)!.5!(i4)+(0, -1)$) {$\Gamma$};
        
          \draw[wire] (i1) to[bend left=15] (G14);
          \draw[wire] (i4) to[out=30,in=0]  (G14);
          \draw[wire] (i2) to[out=-90,in=0]  (G23);
          \draw[wire] (i3) to[bend right=15] (G23);
        \end{tikzpicture}
         \caption{$\{(1,4), (2,3)\}$}
         \label{fig:14-23}
     \end{subfigure}
        \caption{Tensor networks calculating $\Tr(J_{1234}\gamma_{\{1,2,3,4\}} \rho) = J_{1234}(-\Gamma_{1,2}\Gamma_{3,4}+\Gamma_{1,3}\Gamma_{2,4}-\Gamma_{1,4}\Gamma_{2,3})$ (the signs are due to the antisymmetry of the SYK tensor $J$).}
        \label{fig:SYK-tensor-network}
\end{figure}

To calculate the value of each tensor network contraction, we will flatten each of the degree-2 $\Gamma$ tensors into vectors $\Gamma_{vec} \in \mathbb{C}^{4n^2}$, and each will be labeled either ``left'' (covariant) or ``right'' (contravariant), which can be chosen freely. For each of the $k$ tensors $\mathbf{J}$, we will also consider it as a rank $(p, 4-p)$ tensor with $p$ covariant indices and $4-p$ contravariant indices for any $p=0,1,2,3,4$. The rank $(p, 4-p)$ is determined by the number of left vectors the tensor $\mathbf{J}$ is attached to $(4-p)$ and the number of right vectors it is attached to $(p)$. By multi-linearity, each $\mathbf{J}$ tensor thus performs a linear operation equivalent to a $(2n)^p\times(2n)^{4-p}$ matrix denoted $J^{mat}_{p,4-p}$, where $\left(J^{mat}_{p,4-p}\right)_{(i_1,\dots,i_p),(j_1,\dots,J_{4-p})} \triangleq \mathbf{J}_{i_1,\dots,i_p,j_1,\dots,j_{4-p}}$.

The tensor network contracts to give
\begin{equation}
    \beta_\sigma \triangleq \left\langle\pi_L(\Gamma_{vec}^{\otimes N_L})\right| \bigotimes_{p=0}^{4} (J^{mat}_{p,4-p})^{\otimes k_{p,4-p}}\left|\pi_R(\Gamma_{vec}^{\otimes N_R})\right\rangle,
\end{equation}
where $N_L$ is the number of left vectors, $N_R$ is the number of right vectors, $k_{p,4-p}$ is the number of $\mathbf{J}$ tensors of rank $(p,4-p)$, and $\pi_L$ and $\pi_R$ are some permutations on the contravariant and covariant indices of the $\mathbf{J}$ tensors respectively (for an example, see \cref{fig:SYK-tensor-contraction}). Note that $2(N_L + N_R) = 4(k_{0,4}+ k_{1,3}+k_{2,2}+ k_{3,1}+ k_{4,0})=4k$.

\begin{figure}[ht]
     \centering
        \begin{tikzpicture}[x=1cm,y=1cm]
            \node[tensor] (Jtop) at (0, 2.2) {$J_{1,3}^{mat}$};
            \node[tensor] (Jbot) at (0,-2.2) {$J_{3,1}^{mat}$};
            
            \coordinate (T_L)   at ($(Jtop.west)+(0, 0.00)$);
            \coordinate (T_R1)  at ($(Jtop.east)+(0, 0.6)$);
            \coordinate (T_R2)  at ($(Jtop.east)+(0, 0.00)$);
            \coordinate (T_R3)  at ($(Jtop.east)+(0,-0.6)$);
            
            \coordinate (B_L1)  at ($(Jbot.west)+(0, 0.6)$);
            \coordinate (B_L2)  at ($(Jbot.west)+(0, 0.00)$);
            \coordinate (B_L3)  at ($(Jbot.west)+(0,-0.6)$);
            \coordinate (B_R)   at ($(Jbot.east)+(0, 0.00)$);
            
            \draw[leg] (T_L)  -- ++(-0.35,0) coordinate (T_Le);
            \draw[leg] (T_R1) -- ++( 0.35,0) coordinate (T_R1e);
            \draw[leg] (T_R2) -- ++( 0.35,0) coordinate (T_R2e);
            \draw[leg] (T_R3) -- ++( 0.35,0) coordinate (T_R3e);
            
            \draw[leg] (B_L1) -- ++(-0.35,0) coordinate (B_L1e);
            \draw[leg] (B_L2) -- ++(-0.35,0) coordinate (B_L2e);
            \draw[leg] (B_L3) -- ++(-0.35,0) coordinate (B_L3e);
            \draw[leg] (B_R)  -- ++( 0.35,0) coordinate (B_Re);
            
            \coordinate (xL) at (-3,0.0);
            \node[Gnode] (GLu) at ($(xL)+(0,  1.5)$) {$\Gamma$};   
            \node[Gnode] (GLb) at ($(xL)+(0, -1.5)$) {$\Gamma$};   
            
            \draw[wire] (GLu) to[out=20,  in=180] (T_Le);
            \draw[wire] (GLu) to[out=-20, in=180] (B_L1e);
            
            \draw[wire] (GLb) to[out=20,  in=180] (B_L2e);
            \draw[wire] (GLb) to[out=-20, in=180] (B_L3e);
            
            \node[group, fit=(GLu) (GLb)] (VLbox) {};
            \node[lbl] at ($(VLbox.north)+(0,0.20)$) {$v_L$};
            
            \coordinate (xR) at (3,0.0);
            \node[Gnode] (GRu) at ($(xR)+(0,  1.5)$) {$\Gamma$};   
            \node[Gnode] (GRb) at ($(xR)+(0, -1.5)$) {$\Gamma$};   
            
            \draw[wire] (GRu) to[out=160,in=0]  (T_R1e);
            \draw[wire] (GRu) to[out=200,in=0]  (T_R2e);
            
            \draw[wire] (GRb) to[out=160,in=0]  (T_R3e);
            \draw[wire] (GRb) to[out=200,in=0]  (B_Re);
            
            \node[group, fit=(GRu) (GRb)] (VRbox) {};
            \node[lbl] at ($(VRbox.north)+(0,0.20)$) {$v_R$};
        \end{tikzpicture}

    \caption{Tensor network contraction representing $\left\langle\pi_L(\Gamma_{vec}^{\otimes 2})\right| J^{mat}_{1,3}\otimes J^{mat}_{3,1}\left|\pi_R(\Gamma_{vec}^{\otimes 2})\right\rangle$.}
    \label{fig:SYK-tensor-contraction}
\end{figure}

To bound the value of $\beta_\sigma$, we upper-bound the Euclidean norm of the vectors and the associated operator norms of the random matrices $J^{mat}_{p,4-p}$,
\begin{equation}
    \beta_\sigma \leq \left|\langle\pi_L(\Gamma_{vec}^{\otimes N_L})\right|_{2} \left|\pi_R(\Gamma_{vec}^{\otimes N_R})\right|_{2} \cdot \prod_{p=0}^{4} \norm{(J^{mat}_{p,4-p})^{\otimes k_{p,4-p}}}{\mathrm{op}} = \norm{\Gamma}{\mathrm{F}}^{N_L+N_R} \cdot \prod_{p=0}^{4} \norm{J^{mat}_{p,4-p}}{\mathrm{op}}^{k_{p,4-p}}.
\end{equation}
We have $\norm{\Gamma}{\mathrm{F}} = \sqrt{\Tr(\Gamma\Gamma^{\dagger})} = \sqrt{\Tr(\Gamma\Gamma^{\top})} = \sqrt{\Tr(\mathbb{I})} = \sqrt{2n}$ since $\Gamma$ is a real matrix and $\Gamma\Gamma^\top = \mathbb{I}$ for all pure fermionic Gaussian states. 
For Gaussian states that are not pure $\norm{\Gamma}{\mathrm{F}} \le \sqrt{2n}$ by the same argument.
Additionally, $J^{mat}_{p,4-p}$ is a $(2n)^p\times(2n)^{4-p}$ random matrix with i.i.d. standard Gaussian entries in its nonzero indices. Hastings bounds the operator norm of the random Gaussian matrices (e.g., using Theorem 4.4.3 of \cite{vershynin2018high}), giving $\norm{J^{mat}_{0,4}}{\mathrm{op}} \leq O(1)$, $\norm{J^{mat}_{4,0}}{\mathrm{op}} \leq O(1)$, $\norm{J^{mat}_{1,3}}{\mathrm{op}}\leq O(1/\sqrt{n})$, $\norm{J^{mat}_{3,1}}{\mathrm{op}} \leq O(1/\sqrt{n})$, and $\norm{J^{mat}_{2,2}}{\mathrm{op}}\leq O(1/n)$, all with probability exponentially close to 1. This shows that
\begin{align*}
    \beta_\sigma &\leq O(\sqrt{n})^{N_L+N_R} \cdot O(1)^{k_{0,4}+k_{4,0}} \cdot O(1/\sqrt{n})^{k_{1,3}+k_{3,1}}\cdot O(1/n)^{k_{2,2}} \\&\leq O(n)^{k-k_{2,2}-k_{1,3}/2-k_{3,1}/2}O(1)^k.
\end{align*}
Hastings and Ryan O'Donnell show that there always exists some labeling of $\Gamma$ vectors as left and right such that $k_{2,2}=k$, guaranteeing that $\beta_\sigma \leq O(1)^k$ for any pairing $\pi$ \cite{hastings2023fieldtheory}. 
There are at most $(4k-1)!! \leq (4k)^{2k}$ distinct pairings, and summing over all of them gives $\Tr(H^k\rho) \leq (4k)^{2k} \cdot O(1)^{k} = O(k)^{2k}$.

\paragraph{Sparse SYK Generalization}
We prove the following technical lemma generalizing the results of Hastings, which thereby proves \cref{thm:gaussian-upper-bound} upon setting $k=1$:
\begin{lemma}\label{lemma:k-moment-syk-bound}
    In the sparse SYK Hamiltonian $H_{SSYK}$ with any sparsity parameter $p$, with probability at least $1-\mathrm{poly}(n)$ over the draw of the Hamiltonian, the set of all Gaussian states achieves a maximum energy on $H_{SSYK}^k$ which is upper-bounded by $O(k)^{2k}$ when $p \geq \Omega(\log n/n^2)$, and a maximum energy on $H_{SSYK}^k$ which is upper-bounded by $O\left(\frac{k^2\log n}{n\sqrt{p}}\right)^{k}$ when $p < o(\log n/n^2)$. Equivalently
    \begin{equation}
            \Pr\left(\lambda_{\mathrm{Gauss}}(H_{SSYK}^k) \leq \begin{cases}
O(k)^{2k}  & \text{if } p \geq \Omega(\log n/n^2) \\
O\left(\frac{k^2\log n}{n\sqrt{p}}\right)^{k} & \text{if } p < o(\log n/n^2)\end{cases}\right)
 \geq 1-\mathrm{poly}(n).
    \end{equation}
\end{lemma}
We note that the random variable $p^{-1/2}X_{ij}J_{ij}$ has subgaussian norm $O(p^{-1/2})$, so a simple application of the operator norm bounds used by Hastings (e.g., Theorem 4.4.3 of \cite{vershynin2018high}) would imply that $\norm{J^{X}_{2,2}}{\mathrm{op}} \leq 1/\Omega(np^{1/2})$ with probability $\geq 1-\exp(-\Omega(n))$. However, this gives an extra $\Omega(p^{-1/2})$ factor which we wish to avoid in the large $p$ regime. We instead use the following operator norm tail-bound:
\begin{proposition}[Corollary 3.9 of \cite{Bandeira_2016}]\label{lemma:Bandeira}
    Let $X$ be the $m\times m$ symmetric matrix with $X_{ij}$ = $b_{ij}J_{ij}$,
    where $\{J_{ij}: i\geq j\}$ are i.i.d.\ standard Gaussian random variables and $\{b_{ij}: i\geq j\}$ are fixed scalars. Then for any $0<\epsilon < 1/2$, there exists a constant $c_\epsilon$ such that for any $t\ge 0$:
    \begin{equation}
        \Pr\left[\norm{X}{\mathrm{op}} \geq (1+\epsilon)2\sigma + t\right] \leq me^{-t^2/c_\epsilon \sigma_*^2},
    \end{equation}
    where $\sigma \triangleq \max_i \sqrt{\sum_j b_{ij}^2}$ and $\sigma_* \triangleq \max_{ij} |b_{ij}|$.
\end{proposition}
\begin{proof}[Proof of \cref{lemma:k-moment-syk-bound}]
    From \cite{hastings2023fieldtheory} (and as described above), we have $\lambda_{\mathrm{Gauss}}(H^{k}_{SSYK})$ is upper-bounded by $O(k)^{2k} \cdot O\left(n\norm{J^{X}_{2,2}}{\mathrm{op}}\right)^k$, where $J^{X}_{2,2}$ is a $4n^2\times 4n^2$ symmetric matrix with nonzero entries drawn from $\binom{2n}{4}^{-1/2}p^{-1/2}X_{ij}J_{ij}$ where $\{X_{ij}\}_{i\leq j}$ are i.i.d.\ Bernoulli random variables with $\Pr[X_{ij}=1]=p$ and $\{J_{ij}\}_{i\leq j}$ are standard i.i.d.\ Gaussian random variables. We will bound $\norm{J^{X}_{2,2}}{\mathrm{op}}$ with probability $1-\mathrm{poly}(n)$ over the draw of $H_{SSYK}$.

    We note that $J^{X}_{2,2}$ is symmetric as $\mathbf{J}$ is antisymmetric in each index, so we may apply \cref{lemma:k-moment-syk-bound}. For $J^{X}_{2,2}$, $\sigma_* = \binom{2n}{4}^{-1/2}p^{-1/2}$ and $\sigma$ is upper bounded by the maximum over $4n^2$ random variables each with marginal distribution $\sim \binom{2n}{4}^{-1/2}p^{-1/2} \sqrt{\mathrm{Binom}(4n^2, p)}$. 
    
    Consider the two cases:
    \begin{itemize}
        \item When $p \geq \Omega(\log n/n^2)$, by union bound, with probability $1-n^2\exp(-\Omega(n^2p))$, the binomial random variables satisfy $\mathrm{Binom}(n^2, p) \leq Cn^2p$ in all $4n^2$ rows for some constant $C>0$; this is at least $1-1/\mathrm{poly}(n)$. Conditioned on this ``good event'', $\sigma = \binom{2n}{4}^{-1/2}p^{-1/2} \sqrt{Cn^2p} = O(1/n)$. With $p \geq \Omega(\log n/n^2)$, \cref{lemma:Bandeira} gives (conditioned on the good event)
         \begin{equation}
            \Pr\left[\norm{J^{X}_{2,2}}{\mathrm{op}} \geq C(1+\epsilon)/n + t\right] \leq n^2e^{-\Omega(t^2n^4p)} \leq e^{-\Omega(t^2n^2\log n)}.
        \end{equation}
        By setting $t = C'/n$ for some large enough constant $C'>0$, there exists a large enough constant $C''>0$ such that 
             \begin{equation}
            \Pr\left[\norm{J^{X}_{2,2}}{\mathrm{op}} \geq C''/n\right] \leq e^{-(C')^2\log n} \leq 1/\mathrm{poly}(n).
        \end{equation}

        \item When $p < o(\log n/n^2)$, with probability at least $1-1/\mathrm{poly}(n)$, $\mathrm{Binom}(4n^2, p) \leq O(\log n)$ in all $4n^2$ rows. Conditioned on this good event, $\sigma = \binom{2n}{4}^{-1/2}p^{-1/2} O(\log n) = O(\frac{\log n}{n^2\sqrt{p}})$ and \cref{lemma:Bandeira} gives
       \begin{equation}
           \Pr\left[\norm{J^{X}_{2,2}}{\mathrm{op}} \geq \frac{C(1+\epsilon)\log n}{n^2\sqrt{p}} + t\right] \leq n^2e^{-\Omega(t^2n^4p)}.
       \end{equation}
       Setting $t = \frac{C'\log n}{n^2\sqrt{p}}$ for some large enough constant $C'>0$, there exists a large enough constant $C''>0$ such that 
        \begin{equation}
            \Pr\left[\norm{J^{X}_{2,2}}{\mathrm{op}} \geq \frac{C''\log n}{n^2\sqrt{p}}\right] \leq e^{-\frac{(C')^2pn^4 \log^2 n}{pn^4}} \leq 1/\mathrm{poly}(n)
        \end{equation}
    \end{itemize}

    Thus with high probability, $\lambda_{\mathrm{Gauss}}(H^{k}_{SSYK}) \leq O(k)^{2k} O\left(1\right)^k = O(k)^{2k}$ when $p \geq \Omega(\log n / n^2)$, and $\lambda_{\mathrm{Gauss}}(H^{k}_{SSYK}) \leq O(k)^{2k} O\left(\frac{n\log n}{n^2\sqrt{p}}\right)^k = O\left(\frac{k^2\log n}{n\sqrt{p}}\right)^k$ when $p<o(\log n/n^2)$.
\end{proof}

Due to Hastings \cite{hastings2023fieldtheory}, we also get the following as a corollary from \cref{lemma:k-moment-syk-bound} (see \cref{appendix:entangled-gaussian}):
\begin{corollary}\label{cor:magic}
    In the sparse SYK model with any sparsity parameter $p \geq \Omega(\log n/n^2)$, the following variational ansatzes achieve a maximum energy on $H_{SSYK}$ which is upper-bounded by $o(\sqrt{n})$ with probability $\geq 1-\mathrm{poly}(n)$ over the draw of the Hamiltonian.
    \begin{itemize}
        \item Any superposition of $\exp(o(n^{1/4}))$ many orthogonal Gaussian states (i.e., a low-magic state)
        \item The set of states obtained using $o(n^{1/4})/\log n$ steps of the Lanczos algorithm starting from a Gaussian state
    \end{itemize}
\end{corollary}
\subsection{Gaussian state energy lower bound} \label{sec:gaussian_lower}
Recall the sparse SYK model
\begin{equation}
     H_{SSYK}(p) \triangleq \binom{2n}{4}^{-1/2}p^{-1/2}\sum_{\substack{I\subset[2n]\\|I|=4}} J_I X_I \gamma_I,
\end{equation}
where the $J_I$ are i.i.d.\ standard Gaussian random variables and the $X_I$ are i.i.d.\ Bernoulli random variables with $\mathbb{P}\left(X_I=1\right)=p$. We prove that Gaussian states achieve---with probability exponentially close to 1 over the disorder---$\Omega\left(1\right)$ energy for the sparse SYK model for all values $p$, i.e. $\lambda_{\mathrm{Gauss}}(H_{SSYK}) = \Omega(1)$. In addition, when $\omega(1/n^3)<p\leq O(1/n^{2.5})$, we have  $\lambda_{\mathrm{Gauss}}(H_{SSYK}) = \Omega\left(\frac{1}{pn^{2.5}}\right)$ with probability exponentially close to 1.
\begin{theorem}
    In the sparse SYK model with any sparsity parameter $p$, with probability exponentially close to 1 over the draw of the Hamiltonian, there exists a fermionic Gaussian state $\rho$ with energy $\Omega(1)$ for all $p$ and energy $\Omega\left(\frac{1}{pn^{2.5}}\right)$ when $p=\Theta(1/n^{a})$ with $2.5\leq a<3$. Equivalently
    \begin{equation}
            \Pr\left(\lambda_{\mathrm{Gauss}}(H_{SSYK}) \geq \begin{cases}
\Omega(1)  & \text{if } p \geq \Omega(1/n^{2.5}) \\
\Omega\left(\frac{1}{pn^{2.5}}\right) & \text{if } \omega(1/n^3) < p \leq O(1/n^{2.5})\end{cases}\right)
 \geq 1-\exp(-\Omega(n^c))
    \end{equation}
    for some constant $c>0$.
\end{theorem}
\begin{proof}
Considering the range for $p=\Theta(1/n^{a})$, where $2.5\leq a<3$, we note that the algorithm from \cite{Hothem2023} outputs a Gaussian state with approximation ratio $\Omega(1/k)$, where $k$ is the maximum degree of the interaction graph, i.e. the maximum number of terms any single Majorana operator is in. From \cref{thm:ssyk-energy-concentration}, this corresponds to energy $\Omega(\sqrt{n}/k)$. With sparsification parameter $p$, the degree of each vertex is $\Binomial{O(n^3),p}$, which is upper-bounded by $O(n^3p)$ with probability exponentially close to 1 by standard concentration bounds (see \cref{sec:concentration}). Taking a union bound over all $O(n^4)$ vertices, we have $k \leq O(n^3p)$ with probability exponentially close to 1. Conditioned on this event, the energy of the Gaussian state outputted by the algorithm from \cite{Hothem2023} is $\Omega\left(\frac{\sqrt{n}}{n^3p}\right) = \Omega\left(\frac{1}{pn^{2.5}}\right)$.

Our argument for the $\Omega(1)$ energy lower bound is similar to Remark~7.14 in the arXiv version of \cite{hastings2023optimizing}, though we reproduce the full details here for completeness. Note that the Gaussian state depends on the specific SSYK Hamiltonian. Let $\bm{g}^0$ be the $2n\times 2n$ block diagonal matrix:
\begin{equation}
    \bm{g}^0\triangleq i\bm{\varPi}_{n/2}\otimes\bm{Y},
\end{equation}
where
\begin{equation}
    \bm{Y}\triangleq\begin{pmatrix}
        0 & -i\\
        i & 0
    \end{pmatrix}
\end{equation}
is the Pauli $Y$ matrix and $\bm{\varPi}_{n/2}$ is the $n\times n$ projector with ones along the first $n/2$ entries of the diagonal and zeros everywhere else.

Similarly, let $\bm{g}^1$ be the $2n \times 2n$ matrix with entries:
\begin{equation}
    \bm{g}_{k,l}^1\triangleq\bm{1}\left\{k>n\wedge l>n\right\}\left(-1\right)^{\bm{1}\left\{k>l\right\}}\sum_{i<j\leq n} \bm{g}_{i,j}^0 J_{\left\{i,j,k,l\right\}}X_{\left\{i,j,k,l\right\}},
\end{equation}
where $\bm{1}\left\{\cdot\right\}$ denotes the indicator function. The bottom-right $n\times n$ block of $\bm{g}^1$ defines a Wigner matrix with entries with variance $\Theta\left(np\right)$. Standard results in random matrix theory (see \cite{bai1988necessary} or Theorem~2.1.22 of \cite{Anderson_Guionnet_Zeitouni_2009}) then imply that $\norm{\bm{g}^1}{\mathrm{op}}=\Theta\left(n\sqrt{p}\right)$ w.h.p.

Consider now the antisymmetric matrix
\begin{equation}
    \bm{\varSigma}\triangleq\bm{g}^0+C\bm{g}^1,
\end{equation}
where $C=\Theta\left(\frac{1}{n\sqrt{p}}\right)$ such that $\bm{\varSigma}$ defines the covariance matrix of a valid quantum Gaussian state. Using the block diagonal structure of the covariance matrix, the fact that $\Theta\left(n^3\right)$ terms in the Hamiltonian give a nonzero contribution to the energy, and the variance $p$ of the Bernoulli random variables, we calculate the expected energy achieved by this Gaussian state $\rho_{\varSigma}$:
\begin{equation}
    \mathbb{E}\left[\Tr\left(H_{SSYK}(p)\rho_{\varSigma}\right)\right]=C\binom{2n}{4}^{-1/2}p^{-1/2}\times\Theta\left(n^3p\right)=\Theta\left(1\right).
\end{equation}
This statement is also true with probability exponentially close to 1 over the disorder by an identical argument to \cref{sec:universality}, taking the maximum energy over Gaussian states ($\lambda_{\text{Gauss}}$) instead of the maximum energy over all states ($\lambda_{\text{max}}$).
\end{proof}

\subsection{Low circuit complexity energy upper bound}\label{subsec:circuit-complexity-upper}

In this subsection, we prove the following proposition for general low-circuit complexity states:
\begin{proposition}[Sparse SYK circuit complexity lower bound]\label{prop:nlts-upper-bound}
    With probability exponentially close to 1 over the draw of the Hamiltonian, any disorder-independent set of states $S$ containing a state which achieves energy $t$ on the sparse SYK model with sparsity parameter $p$ must have at least $|S| \geq \exp(\Omega(p n^2 t^2))$ distinct states. 
\end{proposition}
\begin{proof}
For the standard $p=1$ SYK-4 model, $\braket{\psi | H_{SYK} | \psi}$ is a Gaussian with variance $\sigma^2 \leq O(n^{-2})$ due to \cref{lem:syk_bound_via_lovasz}. As a reminder, we define our sparse SYK model as
\begin{equation}
    H_{SSYK} \triangleq \binom{2n}{4}^{-1/2}p^{-1/2}\sum_{B\in S^{2n}} J_I X_I B,
\end{equation}
where $X_I$ are independent Bernoulli random variables with $P(X_I=1)=p$. Note that the number of nonzero terms in the sum above is distributed binomially as $\operatorname{Bin}(\binom{2n}{4}, p)$. Therefore by Chernoff bounds on the Binomial distribution (\cref{eqn:simple-binomial-chernoff}), it holds that
\begin{equation}
    \mathbb{P}\left[ \left|\sum_{\substack{I\subset[2n]\\|I|=4}} X_I - p\binom{2n}{4}\right| \geq \delta p\binom{2n}{4} \right] \leq \exp\left(-\Omega\left(\delta^2p\binom{2n}{4}\right)\right).
\end{equation}
Therefore, when $p\binom{2n}{4} = \omega(1)$, there exists a choice of $\delta = o(1)$ which suffices to guarantee that with high probability the number of terms is $\sum_I X_I = \Theta(pn^4)$. 

Since removing a node from $G(S^{2n})$ cannot increase the Lovász theta function of the graph, denoting $G(S^{2n}_p)$ as a given sparsified graph with parameter $p$, we have the universal bound 
\begin{equation}
    \vartheta(G(S^{2n}_p)) \leq \vartheta(G(S^{2n})) \leq O(n^2)
\end{equation}
for all $p$. By \Cref{lem:bound_via_lovasz}, noting that $\sum_I X_I = \Theta(pn^4)$ with probability $1-\exp\left(-\operatorname{\Omega}\left(pn^4\right)\right)$, and standard Gaussian tail bounds, we then have that for any state $\ket{\psi}$:
\begin{equation}
    \Pr(\braket{\psi | H_{SSYK} | \psi} \geq t) \leq \exp(-\Omega(pn^2t^2)).
\end{equation}
Thus, by the union bound, we must have circuit complexity $N(p) \geq \Omega(pn^2t^2)$ to have $\Omega(1)$ probability of achieving energy $t$ over the set of states in the ansatz.
\end{proof}

As shown in \cref{prop:gaussian-net}, we can construct a net of all Gaussian states with $\epsilon$ additive error in energy using $O(n^2 \log (m/\epsilon))$ gates with respect to Hamiltonian $H$, given that $\norm{H}{\mathrm{op}}\leq m$. When $p>\omega(1/n^3)$, we have $\norm{H_{SSYK}}{\mathrm{op}}\leq O(\sqrt{n})$ with probability exponentially close to 1 over the draw of the Hamiltonian (due to \cref{thm:ssyk-energy-concentration}), implying we can use $O(n^2 \log n)$ gates to construct our $\epsilon$-net for any fixed constant $\epsilon$. From \cref{prop:nlts-upper-bound}, for Gaussian states to achieve energy $t$ we must have $\Omega(pn^2t^2) \leq O(n^2\log n)$, or $t \leq O(\sqrt{\log n/p})$. Thus to achieve any constant approximation to the maximum energy of $t = \Omega(\sqrt{n})$, we cannot have $p \geq \Omega(\log n/n)$.

\paragraph{Circuit complexity conjecture.}

\begin{figure}[t]
    \centering
    \includegraphics[]{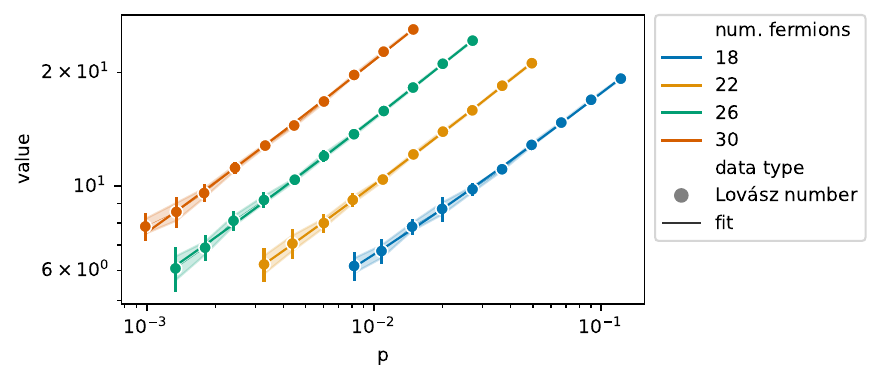}
    \caption{The empirical scaling of the Lovász number of sparsified commutation graphs with respect to the sparsification parameter $p$ conforms well to a scaling of $\Theta(\sqrt{p})$. Here, we plot a comparison of the Lovász number of sparsified commutation graphs of the SYK model with their fit to a square root power law $c_1\sqrt{p}+c_2$ where $c_1$ and $c_2$ are the fit parameters. Error bars are equal to the standard deviation over values drawn from $10$ random sparsifications of the initial graph for each data point.}
    \label{fig:lovasz_numerics}
\end{figure}

\Cref{prop:nlts-upper-bound} states that the circuit complexity of an ansatz achieving energy $t$ on the sparse SYK model with sparsity parameter $p$ must have circuit complexity $\Omega(pn^2t^2)$. We conjecture here that the stronger lower-bound of $\Omega(\sqrt{p}n^2 t^2)$ also holds:

\begin{conjecture}[Sparse SYK circuit complexity conjecture]\label{conj:nlts-upper-bound}
    With probability exponentially close to 1 over the draw of the Hamiltonian, any disorder-independent set of states $S$ containing a state which achieves energy $t$ with high probability on the sparse SYK model with sparsity parameter $p$ must have at least $|S| \geq \exp(\Omega(\sqrt{p} n^2 t^2))$ distinct states.
\end{conjecture}

\Cref{fig:lovasz_numerics} provides numerical evidence for \Cref{conj:nlts-upper-bound} for various small system sizes. There, we plot the value of the Lovász number of sparsified graphs for various values of $p$. The data conforms well to a power law fit governed by the equation $c_1\sqrt{p}+c_2$ where $c_1, c_2$ are the fit parameters.\footnote{The code to reproduce the numerical experiments in \Cref{subsec:circuit-complexity-upper} is available at \url{https://github.com/bkiani/Lovasz_for_SYK}.} By \Cref{lem:bound_via_lovasz}, this scaling would imply that the commutation index $\Delta$ of the sparsified model scales as $\Delta=O(p^{-1/2}n^2)$, resulting in the statement of \Cref{conj:nlts-upper-bound}.

\begin{remark}
    For continuous ansatzes with $d$ parameters (e.g., fermionic Gaussian states for $d=n^2$), we can generally construct an $\epsilon$-net using $d\cdot \text{polylog}(d)$ circuit gates (e.g., see \cref{prop:gaussian-net}). Therefore \cref{prop:nlts-upper-bound} and \cref{conj:nlts-upper-bound} correspond to lower-bounds of $d \gtrsim p n^2 t^2$ and $d\gtrsim \sqrt{p} n^2 t^2$ parameters respectively for any ansatz with a state achieving at least $t$ energy, up to polylogarithmic factors in $n$. 
\end{remark}
\section{Optimal energy quantum algorithm for sparse SYK Hamiltonians}\label{sec:sparse-ho-alg}
In this section we prove \cref{thm:sparse-ho}, which we restate below: \quantumthm*

First, consider $2n$ Majorana operators split into two sets, $\left\{\phi_1, \phi_2,\ldots,\phi_{n_1}\right\}$ of cardinality $n_1$ and $\left\{\chi_1, \chi_2,\ldots,\chi_{n_2}\right\}$ of cardinality $n_2$. We define a new \emph{sparse} two-color SYK-4 Model:
\begin{definition}[Sparse two-color SYK model]
    Let $H^{(2)}_{SSYK}$ be the Hamiltonian defined as:
    \begin{equation}
        H^{(2)}_{SSYK}(p) \triangleq \frac{i}{\sqrt{n_2}} \sum_{j=1}^{n_2}\tau_j(p) \chi_j,
        \qquad \qquad
        \tau_j(p) \triangleq i \binom{n_1}{3}^{-1/2} p^{-1/2} \sum_{S \subset [n_1], |S| = 3} J_{S,j} X_{S,j} \phi_S,
    \end{equation}
    where the $X_{S,j}$ are i.i.d.\ Bernoulli random variables with $\Pr(X_{S,j} = 1) = p$, $J_{S,j}$ are i.i.d. standard Gaussian random variables, $\phi_S \triangleq \prod_{i \in S} \phi_i$ for the ordered set $S$, and $n_1, n_2 = \Theta(n)$. We can see that this directly generalizes the original two-color (dense) SYK-4 model \cite{hastings2023optimizing} by setting $p = 1$. Note that we will write $H^{(2)}_{SSYK}$ and $\tau_j$ where the dependence on $p$ is implicitly assumed.
\end{definition}

In \cref{subsec:reduction} we show that near-optimal states of the sparse two-color SYK model are (with high probability) near-optimal states of the standard sparse SYK model. Then in \cref{subsec:ho_algorithm}, we show that the Hastings--O'Donnell algorithm (\cref{alg:generalized_HO}) achieves $\Omega(\sqrt{n})$ energy on the sparse two-color SYK model.

\subsection{Reduction from sparse SYK-4 model}\label{subsec:reduction}
Consider the sparse SYK-4 model
\begin{equation}\label{eq:sparse_syk_ham_sec_4}
    H_{SSYK}(p) \triangleq \binom{2n}{4}^{-1/2} p^{-1/2} \sum_{\substack{I \subset [2n]\\|I| = 4}} X_I J_I \gamma_I
\end{equation}
where $\{X_I\}_I$ are i.i.d\ Bernoulli random variables with $\Pr(X_I = 1) = p$ and $\{J_I\}_I$ are i.i.d\ standard Gaussians. We prove the following reduction:
\begin{lemma}\label{lemma:two-color-reduction}
    Fix some sparsification parameter $p$. Let $H^{(2)}_{SSYK}$ be a sparse two-color SYK Hamiltonian with $0.1n \leq n_1, n_2 \leq 2n$. For fixed state $\rho$, if $\Tr(H^{(2)}_{SSYK}\rho) \geq C_1\sqrt{n}$ with probability $e^{-\Omega(n)}$ for some constant $C_1$, then $\Tr(H_{SSYK}\rho) \geq C_2\sqrt{n}$ with probability $e^{-\Omega(n)}$ for some constant $C_2$.
\end{lemma}
\begin{proof}
We partition the Majorana operators into two sets, indexed by $A = [n_1]$ and $B = [2n] \setminus A$. Let $n_2 = 2n - n_1$, and assume $n_2 \leq n_1$  without loss of generality. Consider all Majorana operators $\gamma_I$ in the sparse SYK Hamiltonian \cref{eq:sparse_syk_ham_sec_4}; we let $T$ denote the subset of all $\gamma_I$ for which $I=\{i_1 < i_2 < i_3 < i_4\}$ with $i_1, i_2, i_3 \in A$ and $i_4 \in B$. We write:
\begin{equation}
    H_{SSYK}=H_T+H_{\overline{T}},
\end{equation}
where $H_T$ contains only the terms with $\gamma_I\in T$ and $H_{\overline{T}}$ the terms with $\gamma_I\in \overline{T}$. We can rewrite $H_T$ as
\begin{align*}
    H_T &= \binom{2n}{4}^{-1/2} p^{-1/2} \sum_{j \in B}  \left(\sum_{I \subset A, |I| = 3} X_{I,j} J_{I,j} \gamma_I\right) \cdot \gamma_j\\
    &= \frac{-2i}{\sqrt{2n}} \sum_{j \in B} \left(i \binom{2n-1}{3}^{-1/2}p^{-1/2} \sum_{I \subset [n_1], |I| = 3} X_{I,j} J_{I,j} \gamma_I\right) \cdot \gamma_j\\
    &= \frac{2\sqrt{n_2}}{\sqrt{2n}}\cdot\sqrt{\frac{\binom{n_1}{3}}{\binom{2n-1}{3}}}\frac{i}{\sqrt{n_2}} \sum_{j \in B} \tau_j(p) \cdot \gamma_j\\
    &= 2\sqrt{\frac{n_2\binom{n_1}{3}}{2n\binom{2n-1}{3}}} H^{(2)}_{SSYK}(p).
\end{align*}

Since both $n_1, n_2 \in [0.1n, 2n]$, we have $H_T(p) = C_3 H^{(2)}_{SSYK}(p)$, for some absolute constant $C_3$. Thus if some state $\rho$ has energy $\Omega(\sqrt{n})$ (independent of $p$) on Hamiltonian $H^{(2)}_{SSYK}(p)$, it also has energy $\Omega(\sqrt{n})$ on $H_T$.

Now we consider the $H_{\overline{T}}$ portion of the Hamiltonian with associated energy contribution
\begin{equation}
    \Tr(H_{\overline{T}}\rho) = \binom{2n}{4}^{-1/2}p^{-1/2} \sum_{I \in \overline{T}} X_IJ_I \Tr(\gamma_I\rho),
\end{equation}
where $|\Tr(\gamma_I\rho)| \leq 1$ and $\{J_I\}_I$ are i.i.d.\ standard Gaussian random variables. Since the number of non-zero terms in the summation is upper-bounded by $m \sim \Binomial{\binom{2n}{4}, p}$, by a Chernoff bound (\cref{eqn:simple-binomial-chernoff}) we have that $\Pr[m \geq Cpn^4] \leq e^{-\Omega(n)}$. Conditioned on the complement of this event which occurs with probability exponentially close to 1, $\Tr(H_{\overline{T}}\rho)$ is a Gaussian random variable with zero mean and variance $\leq O(1)$. By another Chernoff bound we have $|\Tr(H_{\overline{T}}\rho)| \leq c\sqrt{n}$ with probability exponentially close to 1 for any constant $c$.

Thus if $\Tr(H^{(2)}_{SSYK}\rho) \geq C_1\sqrt{n}$ with high probability, then $\Tr(H_{SSYK}\rho) \geq C_2\sqrt{n}$ with high probability. This gives a reduction from sparse SYK to sparse two-color SYK with at most a constant factor difference in energy scaling.
\end{proof}

\subsection{Hastings--O'Donnell algorithm applied to sparse SYK}\label{subsec:ho_algorithm}
The Hastings--O'Donnell Algorithm \cite{hastings2023optimizing} is a quantum algorithm for finding a constant approximation to the optimal energy state for the two-color SYK Hamiltonian. We show that this algorithm also works on the sparse two-color SYK Hamiltonian when $p \geq \Omega(\log n/ n)$.

Let $\sigma_j$ for $j=1, \dots, n_2$ denote an additional set of $n_2$ Majorana operators satisfying the Clifford algebra relations over degree $2n_2 + n_1$ with both themselves and additionally the Majoranas $\{\chi_j\}_j$ and $\{\phi_j\}_j$. Construct $H'$ to be the quadratic fermionic Hamiltonian
\begin{equation}
    H' \triangleq \frac{i}{\sqrt{n_2}} \sum_{j=1}^{n_2} \sigma_j \chi_j,
\end{equation}
which is optimized by the Gaussian state
\begin{equation}
    \rho_0 \triangleq \frac{1}{2^{n_2 + n_1 / 2}} \prod_{j=1}^{n_2} (\mathbb{I}+i\sigma_j \chi_j).
\end{equation}
Let $\mathrm{Ad}_U(g) \triangleq UgU^{-1}$. Define the transformed (non-Gaussian) state $\rho_\theta$, for a given parameter $\theta \in \mathbb{R}$:
\begin{equation}
    \rho_\theta \triangleq \mathrm{Ad}_{e^{-\theta\zeta}}(\rho_0) = e^{-\theta\zeta}\rho_0e^{+\theta\zeta}, \qquad \zeta \triangleq \sum_{j=1}^{n_2} \tau_j \sigma_j.
\end{equation}
In \cite{hastings2023optimizing} it was shown that this state achieves with high probability a constant approximation ratio for the dense ($p=1$) two-color SYK model for some constant $\theta$. This algorithm is summarized as \cref{alg:generalized_HO} and a visual diagram is shown in \cref{fig:HO_algorithm}.
\begin{remark}
    Technically since $H_{SSYK}^{(2)}$ is only defined on the $\{\chi_j\}_j, \{\phi_j\}_j$ Majorana operators, given our algorithm output we trace out the subsystem defined on the additional ancilla $\{\sigma_j\}_j$ Majoranas. However this does not change our analysis and we ignore it for simplicity.
\end{remark}
\begin{algorithm} 
    \SetKwInOut{Input}{Input}
    \SetKwInOut{Output}{Output}
    
    \Input{$H_{SSYK}^{(2)}$ randomly drawn from sparse two-color SYK distribution with parameter $p = \Omega(\log n/n)$.}
    \medskip
    
    $\rho_0 \gets \frac{1}{2^{n_2 + n_1 / 2}} \prod_{j=1}^{n_2} (\mathbb{I}+i\sigma_j \chi_j)$
    \medskip
    
    $\zeta \gets \sum_{j=1}^{n_2} \tau_j \sigma_j$
    \medskip

    $\theta\gets\text{constant value from \Cref{lemma:ho-energy}}$
    
    $\rho_\theta \gets e^{-\theta\zeta}\rho_0e^{+\theta\zeta}$
    \medskip
    
    \Return{$\rho_\theta$}
    \caption{Hastings--O'Donnell for Sparse SYK}
    \label{alg:generalized_HO}
\end{algorithm}

\begin{figure}[ht]
    \centering
    \usetikzlibrary{arrows.meta,calc,fit,decorations.pathreplacing}
    \tikzset{
      chi/.style={circle, draw, thick, fill=red!10,  minimum size=7.0mm, inner sep=0pt},
      sig/.style={circle, draw, thick, fill=gray!10, minimum size=7.0mm, inner sep=0pt},
      tauOuter/.style={circle, draw, thick, fill=blue!10, minimum size=15.0mm, inner sep=0pt},
      tauInner/.style={circle, draw, thick, fill=blue!3,  minimum size=3.8mm, inner sep=0pt},
      couple/.style={thick},
      pair/.style={thick, dashed},
      dashedbox/.style={draw, thick, dashed, rounded corners},
      inneredge/.style={thick},
      tautext/.style={font=\small, inner sep=1pt, anchor=east},
      linelabel/.style={font=\small, fill=white, inner sep=1pt},
      brace/.style={decorate, decoration={brace, amplitude=6pt}, thick},
      unitarrow/.style={-Latex, line width=1.2pt, >={Latex[length=3.2mm,width=2.6mm]}},
    }

    \begin{tikzpicture}[font=\small, x=1cm, y=1cm]
    
    \def\yOne{4.2}
    \def\yTwo{2.4}
    \def\yThr{0.6}
    \def\yFor{-1.2}
    
    \def\xTau{2.2}
    \def\xChi{4.8}
    \def\xSig{7.4}
    
    \newcommand{\TauNode}[4]{%
      \node[tauOuter] (#1) at (#2,#3) {};
      \node[tautext] (#1lbl) at ($(#1.west)+(-0.20,0)$) {$\tau_{#4}\sim \text{SSYK-3}$};
    
      \coordinate (#1c) at (#1.center);
      \node[tauInner] (#1i1) at ($(#1c)+( 0.36, 0.20)$) {};
      \node[tauInner] (#1i2) at ($(#1c)+(-0.34, 0.20)$) {};
      \node[tauInner] (#1i3) at ($(#1c)+( 0.00,-0.34)$) {};
    
      \draw[inneredge] (#1i1.center) -- (#1i2.center);
      \draw[inneredge] (#1i2.center) -- (#1i3.center);
      \draw[inneredge] (#1i3.center) -- (#1i1.center);
    }
    
    \TauNode{tau1}{\xTau}{\yOne}{1}
    \TauNode{tau2}{\xTau}{\yTwo}{2}
    \TauNode{tau3}{\xTau}{\yThr}{3}
    \TauNode{tau4}{\xTau}{\yFor}{4}
    
    \node[chi] (chi1) at (\xChi,\yOne) {$\chi_{1}$};
    \node[chi] (chi2) at (\xChi,\yTwo) {$\chi_{2}$};
    \node[chi] (chi3) at (\xChi,\yThr) {$\chi_{3}$};
    \node[chi] (chi4) at (\xChi,\yFor) {$\chi_{4}$};
    
    \node[sig] (sig1) at (\xSig,\yOne) {$\sigma_{1}$};
    \node[sig] (sig2) at (\xSig,\yTwo) {$\sigma_{2}$};
    \node[sig] (sig3) at (\xSig,\yThr) {$\sigma_{3}$};
    \node[sig] (sig4) at (\xSig,\yFor) {$\sigma_{4}$};
    
    \draw[couple] (tau1.east) -- (chi1.west);
    \node[linelabel] at ($(tau1.east)!0.5!(chi1.west)$) {$\tau_1\chi_1$};
    
    \draw[couple] (tau2.east) -- (chi2.west);
    \node[linelabel] at ($(tau2.east)!0.5!(chi2.west)$) {$\tau_2\chi_2$};
    
    \draw[couple] (tau3.east) -- (chi3.west);
    \node[linelabel] at ($(tau3.east)!0.5!(chi3.west)$) {$\tau_3\chi_3$};
    
    \draw[couple] (tau4.east) -- (chi4.west);
    \node[linelabel] at ($(tau4.east)!0.5!(chi4.west)$) {$\tau_4\chi_4$};
    
    \draw[pair] (chi1.east) -- (sig1.west);
    \node[linelabel] at ($(chi1.east)!0.5!(sig1.west)$) {$\sigma_1\chi_1$};
    
    \draw[pair] (chi2.east) -- (sig2.west);
    \node[linelabel] at ($(chi2.east)!0.5!(sig2.west)$) {$\sigma_2\chi_2$};
    
    \draw[pair] (chi3.east) -- (sig3.west);
    \node[linelabel] at ($(chi3.east)!0.5!(sig3.west)$) {$\sigma_3\chi_3$};
    
    \draw[pair] (chi4.east) -- (sig4.west);
    \node[linelabel] at ($(chi4.east)!0.5!(sig4.west)$) {$\sigma_4\chi_4$};
    
    \node[inner sep=0pt, fit=(chi1)(chi4)(sig1)(sig4)] (gaussfit) {};
    
    \node[dashedbox, inner xsep=4pt, inner ysep=7pt,
          fit=(tau1lbl)(tau4lbl)(tau1)(tau4)(chi1)(chi4)] (twocolbox) {};
    
    \node[font=\small, align=center, anchor=south] at ($(twocolbox.north)+(0,0)$) {%
    Two-color SSYK\\
    $\displaystyle H^{(2)}_{\mathrm{SSYK}}\propto \sum_{j=1}^{n_2}\tau_j\,\chi_j$
    };
    
    \draw[brace] ($(gaussfit.north east)+(0.2,0)$) -- ($(gaussfit.south east)+(0.2,0)$)
      node[midway, xshift=2cm, align=center] {%
    Gaussian state\\
    $\displaystyle \rho_0 \propto \prod_{j=1}^{n_2}\bigl(\mathbb{I}+ i\sigma_j\chi_j\bigr)$
    };

    \coordinate (Astart) at ($(sig4.south west)+(-0.2,-0.2)$);
    \coordinate (Aend)   at ($(tau4.south east)+( 0.3,-0.2)$);

    \draw[unitarrow, bend left=30]
      (Astart) to
      node[midway, below, fill=white, inner sep=1pt] {Hamiltonian evolution: $\displaystyle H = i\zeta = i\sum_{j=1}^{n_2} \tau_j \sigma_j$}
      (Aend);
    
    \end{tikzpicture}
    \caption{Visualization of the Hastings--O'Donnell algorithm on sparse SYK (\cref{alg:generalized_HO}).}
    \label{fig:HO_algorithm}
\end{figure}

We prove the following lemma, which shows the algorithm also works for the sparse two-color SYK model with sufficiently large $p$:
\begin{lemma}\label{lemma:ho-energy}
    For any $p \geq \Omega(\log n/n)$, there exist positive constants $\theta, C \in \R$ such that $\Tr(\rho_\theta H_{SSYK}^{(2)}) \geq C\sqrt{n}$ with probability at least $1-e^{-\Omega(n)}$ over the draw of $H_{SSYK}^{(2)}$. 
\end{lemma}
\begin{proof}
Following the high-level idea from \cite{hastings2023optimizing}, we calculate the energy of $\rho_\theta$ on $H_{SSYK}^{(2)}$.
\begin{equation}
    \Tr\left(\rho_\theta H_{SSYK}^{(2)}\right) = \Tr\left(e^{-\theta\zeta}\rho_0e^{+\theta\zeta} H_{SSYK}^{(2)}\right) = \Tr\left(\rho_0 \mathrm{Ad}_{e^{+\theta\zeta}}(H_{SSYK}^{(2)})\right)
\end{equation}
Using the Baker--Campbell--Hausdorff formula,
\begin{equation}
    \mathrm{Ad}_{e^{+\theta\zeta}}(H_{SSYK}^{(2)}) = H_{SSYK}^{(2)} + \theta[\zeta, H_{SSYK}^{(2)}] + \theta^2\int_{0}^{1} (1-x) \mathrm{Ad}_{e^{+x\theta\zeta}}\left([\zeta, [\zeta, (H_{SSYK}^{(2)})]]\right)\,dx,
\end{equation}
and thus
\begin{align*}
    \Tr\left(\rho_\theta H_{SSYK}^{(2)}\right) &= \Tr\left(H_{SSYK}^{(2)}\right) + \theta\Tr\left(\rho_0 [\zeta, H_{SSYK}^{(2)}]\right) + \\&\qquad\theta^2\int_{0}^{1} (1-x) \Tr\left(\rho_0\mathrm{Ad}_{e^{+x\theta\zeta}}\left([\zeta, [\zeta, (H_{SSYK}^{(2)})]]\right)\right)\,dx\\ &= \theta\Tr\left(\rho_0 [\zeta, H_{SSYK}^{(2)}]\right) + \theta^2\int_{0}^{1} (1-x) \Tr\left(\rho_0[\zeta, [\zeta, (H_{SSYK}^{(2)})]]\right)\,dx.
\end{align*}
We finally get
\begin{equation}\label{eqn:main-ho}
    \Tr\left(\rho_\theta H_{SSYK}^{(2)}\right) \geq \theta\Tr\left(\rho_0 [\zeta, H_{SSYK}^{(2)}]\right) - \theta^2 \norm{[\zeta, [\zeta, (H_{SSYK}^{(2)})]]}{\mathrm{op}};
\end{equation}
therefore we wish to lower bound the right-hand side of \cref{eqn:main-ho} by $\Omega(\sqrt{n})$ for some $\theta$. We show $\Tr(\rho_0 [\zeta, H_{SSYK}^{(2)}]) \geq C_1\sqrt{n}$ for some constant $C_1>0$ in \cref{lemma:single-commutator} and $\norm{[\zeta, [\zeta, (H_{SSYK}^{(2)})]]}{\mathrm{op}} \leq C_2 \sqrt{n}$ for some constant $C_2>0$ in \cref{lemma:double-commutator}, both with probability at least $1-e^{-\Omega(n)}$ over the draw of $H_{SSYK}^{(2)}$. Together this shows that there exists a constant $\theta$ such that $\Tr(\rho_\theta H_{SSYK}^{(2)}) \geq C\sqrt{n}$.
\end{proof}

\paragraph{Single-commutator term.}
The following lemma lower-bounds the single-commutator term in \cref{eqn:main-ho}:
\begin{lemma}\label{lemma:single-commutator}
    For all $p \geq \Omega(1/n^2)$, $\Tr(\rho_0 [\zeta, H_{SSYK}^{(2)}]) \geq C\sqrt{n}$ with probability at least $1-e^{-\Omega(n)}$ over the draw of $H_{SSYK}^{(2)}$. 
\end{lemma}
\begin{proof}
We begin by directly calculating out the trace term similar to the dense SYK-4 case.
\begin{align*}
    \Tr(\rho_0 [\zeta, H_{SSYK}^{(2)}]) &= \frac{i}{\sqrt{n_2}} \sum_{j=1}^{n_2} \Tr(\rho_0 [\zeta, \tau_j \chi_j]) \\ &= \frac{i}{\sqrt{n_2}} \sum_{j,k=1}^{n_2} \Tr(\rho_0 [\tau_k \sigma_k, \tau_j \chi_j]) \\ &= \frac{i}{\sqrt{n_2}} \sum_{j=1}^{n_2} \Tr(\rho_0 [\tau_j \sigma_j, \tau_j \chi_j]) \\ &= \frac{2i}{\sqrt{n_2}} \sum_{j=1}^{n_2} \Tr(\rho_0 \sigma_j \chi_j \tau_j^2) \\ &= \frac{2}{\sqrt{n_2}}\cdot 2^{-(n_2 + n_1/2)} \sum_{j=1}^{n_2} \Tr(\tau_j^2)
\end{align*}
From here our proof differs from the dense SYK-4 case, and we use our new definition of $\{\tau_j\}_j$ to get
\begin{align*}
   \frac{2}{\sqrt{n_2}}  \cdot 2^{-(n_2 + n_1/2)}\sum_{j=1}^{n_2} \Tr(\tau_j^2) &=  \frac{2}{\sqrt{n_2}} \cdot 2^{-(n_2 + n_1/2)}\sum_{j=1}^{n_2} \binom{n_1}{3}^{-1} p^{-1} \Tr\left((\sum_{\substack{S\subseteq[n_1]\\|S|=3}} X_{S} J_{S, j} \phi_S)^2\right)\\
    &= \frac{2}{\sqrt{n_2}} \binom{n_1}{3}^{-1} p^{-1} \sum_{j=1}^{n_2} \sum_{\substack{S\subseteq[n_1]\\|S|=3}} X_{S} J_{S, j}^2
\end{align*}
Conditioned on the $X_S$, $\sum_S X_{S} J_{S, j}^2$ is a chi-squared random variable with degrees of freedom $d$, where $d \leq pn_1^3$ with high probability over the draw of the Binomial random variables $X_S$. The Laurent--Massart \cite{laurent-massart} upper-bound for chi-squared random variables gives:
\begin{equation}
    \Pr(\chi^2 \leq k - 2\sqrt{kx}) \leq \exp{(-x)},
\end{equation}
where $k$ is the degrees of freedom of $\chi^2$. Plugging in $k = pn_1^3$ and $x = k/16$ gives:
\begin{equation}
    \Pr\left[\sum_{\substack{S\subseteq[n_1]\\|S|=3}} X_{S} J_{S, j}^2 \leq \frac{1}{2}pn_1^3\right] \leq \Pr[d \leq pn_1^3] \Pr\left[\sum_{i = 1}^{pn_1^3} J_{S_i, j}^2 \leq \frac{1}{2}pn_1^3\right] +\exp\left(-\operatorname{\Omega}\left(n\right)\right)\leq \exp{(-\Omega(n))}
\end{equation}
when $p = \Omega(1/n^2)$ and $n_1=\Theta(n)$. Applying a union bound over the $j$, we have:
\begin{equation}
    \Tr(\rho_0 [\zeta, H_{SSYK}^{(2)}(p)]) = \frac{2}{\sqrt{n_2}}\binom{n_1}{3}^{-1}p^{-1}\sum_{j=1}^{n_2}  \sum_{\substack{S\subseteq[n_1]\\|S|=3}} X_{S} J_{S, j}^2 \geq C\sqrt{n}
\end{equation}
with probability exponentially close to $1$ assuming $p = \Omega(1/n^2)$ and $n_2 = \Theta(n)$.
\end{proof}

\paragraph{Double-commutator term.}
The following lemma upper-bounds the double-commutator term in \cref{eqn:main-ho}:
\begin{lemma}\label{lemma:double-commutator}
    For all $p \geq \Omega(\log n/n)$, $\norm{[\zeta, [\zeta, (H_{SSYK}^{(2)})]]}{\mathrm{op}} \leq C_2 \sqrt{n}$ with probability at least $1-e^{-\Omega(n)}$ over the draw of $H_{SSYK}^{(2)}$. 
\end{lemma}
\begin{proof}
We define $J_{S, j}$ as a mean-zero Gaussian with variance 1 if $S$ has all distinct elements, and variance 0 otherwise. A direct calculation shown in \cite{Herasymenko_2023} gives the value of $[\zeta, [\zeta, H_{SSYK}^{(2)}]]$ as
\begin{equation}
    \begin{split}
        A \triangleq \frac{1}{\sqrt{n_2}}\binom{n_1}{3}^{-3/2}p^{-3/2}\sum_{j, k, \ell=1}^{n_2}\sum_{\substack{S, S', S'' \subset [n_1]\\|S|,|S'|,|S''|=3}} &f(S,S',S'',j,k,\ell) (X_{S, j} J_{S, j})(X_{S', k} J_{S', k}) \\
        &\times (J_{S'', \ell} X_{S'', \ell})\Big(\phi_S\sigma_j\phi_{S'}\sigma_k\phi_{S''}\chi_\ell\Big)
    \end{split}
\end{equation}
where $f(S,S',S'',j,k,\ell) \triangleq 1$ if $(|S''\cap S| \text{ is odd})\wedge(|S'\cap (S''\triangle S)|+\delta_{k,\ell} \text{ is odd})$ and 0 otherwise. Thus $[\zeta, [\zeta, H_{SSYK}^{(2)}]]$ is a polynomial over the anticommuting characters $\{\phi_j\}_j, \{\sigma_j\}_j, \{\chi_j\}_j$, which we denote as $A$. We show that it is sufficient to bound the $2n^{th}$-moment of this polynomial when the anticommuting characters are replaced with 1, which we denote as $A(1)$. 

Assume the following lemma (which we prove in \cref{sec:moment-proofs}),
\begin{lemma} \label{lemma:moment_bound}
    Let 
    \begin{align*}
         A(1) \triangleq \frac{1}{\sqrt{n_2}}\binom{n_1}{3}^{-3/2}p^{-3/2}\sum_{j, k, \ell=1}^{n_2}\sum_{\substack{S, S', S'' \subset [n_1]\\|S|,|S'|,|S''|=3}} &f(S,S',S'',j,k,\ell) (X_{S, j} J_{S, j})\\&\times (X_{S', k} J_{S', k}) (X_{S'', \ell} J_{S'', \ell}),
    \end{align*}
    where $f(S,S',S'',j,k,\ell) \triangleq 1$ if $(|S''\cap S| \text{ is odd})\wedge(|S'\cap (S''\triangle S)|+\delta_{k,\ell} \text{ is odd})$ and 0 otherwise. We have the following moment bound:
    \begin{equation} \label{gaussian_bound}
       \norm{A(1)}{2n} \leq C\sqrt{n}.
    \end{equation}  
\end{lemma}
We prove that \cref{lemma:moment_bound} implies \cref{lemma:double-commutator}, following the moment method techniques outlined in \cite{hastings2023optimizing, Herasymenko_2023}. By Jensen's inequality we have
\begin{equation}
    \E\left[\norm{A}{\mathrm{op}}\right]\leq  \E\left[\norm{A}{\mathrm{op}}^{2n}\right]^{1/2n}
\end{equation}
and
\begin{equation}
   \E\left[\norm{A}{\mathrm{op}}^{2n}\right]^{1/2n} \leq \E\left[\Tr\left(A^{2n}\right)\right]^{1/2n}\leq 
   \left((2^{n_2 + n_1/2} \E\left[A(1)^{2n}\right]\right)^{1/2n}
\end{equation}
where the last inequality is due to the fact that all Majoranas have operator norm bounded by 1, therefore the trace of a product of Majoranas is at most their dimension $2^{n_2 + n_1/2} \leq 2^{2n}$. 

Therefore, we get
\begin{equation}
     \E\left[\norm{A}{\mathrm{op}}\right] \leq 2 \norm{A(1)}{2n} \leq 2C\sqrt{n}
\end{equation}
for some constant $C$, assuming \cref{lemma:moment_bound}. By Markov's inequality,
\begin{equation}
    \Pr\left(\norm{A}{\mathrm{op}} \geq 4C\sqrt{n}\right) \leq \frac{\E\left[\norm{A}{\mathrm{op}}^{2n}\right]}{(4C\sqrt{n})^{2n}} \leq \left(\frac{1}{2}\right)^{2n},
\end{equation}
which proves \cref{lemma:double-commutator}.
\end{proof}

\subsection{Proof of main theorem}
We complete the proof of \cref{thm:sparse-ho}:
\begin{proof}[Proof of \cref{thm:sparse-ho}]
    In \cref{lemma:two-color-reduction}, we show that finding a $\Omega(\sqrt{n})$ energy state of sparse SYK-4 reduces to finding a $\Omega(\sqrt{n})$ energy state of two-color sparse SYK-4. We show that the Hastings--O'Donnell quantum algorithm outputs a $\Omega(\sqrt{n})$ energy state of the two-color sparse SYK-4 model when $p \geq \Omega(\log n/n)$ with probability at least $1-e^{-\Omega(n)}$ over the draw of the Hamiltonian in \cref{lemma:ho-energy}. Together, this implies that Hastings--O'Donnell achieves $\Omega(\sqrt{n})$ energy on the sparse SYK model, as desired.
\end{proof}

\section*{Acknowledgments}
This work was primarily done while M.D. was visiting Caltech through the Summer Undergraduate Research Fellowship (SURF), funded in part by the Institute for Quantum Information and Matter at Caltech, an NSF Physics Frontiers Center (PHY-2317110). They would like to thank John Preskill and the Caltech SURF program for their continued and invaluable support. 

M.D. and R.K. acknowledge funding from the U.S. Department of Energy, Office of Science, National Quantum Information Science Research Center, Quantum Systems Accelerator (DE-SCL0000121). E.R.A. was funded in part by the Walter Burke Institute for Theoretical Physics at Caltech.

The authors would also like to thank multiple anonymous reviewers for their helpful comments and suggestions.

\newpage
\printbibliography
\newpage
\appendix
\section{Improved maximum eigenvalue upper bound}\label[appendix]{appendix:upper-bound}
In this section we show a stronger upper bound on the maximum eigenvalue of sparse SYK Hamiltonians. \cite{hastings2023optimizing} showed that the maximum eigenvalue of the SYK-4 model can be upper-bounded by $O(\sqrt{n})$ with probability exponentially close to 1 using a straightforward application of the matrix Chernoff bound, given as follows:
\begin{proposition}[Theorem 1.2 of \cite{Tropp_2011}]
   Consider a finite sequence $\{A_k\}$ of fixed, self-adjoint matrices with dimension $d$, and let $J_k$ be a finite sequence of independent standard
normal random variables. Then, for all $t \geq 0$,
\begin{equation}
    \Pr\left(\lambda_{\mathrm{max}}\left(\sum_k J_k A_k\right) \geq t\right) \leq d \cdot e^{-\frac{t^2}{2\sigma^2}}
\end{equation}
where $\sigma^2 \triangleq \norm{\sum_k A_k^2}{\mathrm{op}}$. 
\end{proposition}

We reproduce the proof of \cite{hastings2023optimizing} for the dense SYK model for completeness. Since $\gamma_I$ in the Weyl--Brauer representation is a $2^{n}$ dimensional matrix and $\norm{\sum_I \gamma_I^2}{\mathrm{op}} = \norm{\sum_I \mathbb{I}}{\mathrm{op}} = \binom{2n}{4}$, we have:
\begin{align*}
    \Pr\left(\lambda_{\mathrm{max}}(H_{SYK}) \geq C\sqrt{n}\right) &\leq  \Pr\left(\lambda_{\mathrm{max}}\left(\sum_I J_I \gamma_I\right) \geq Cn^{5/2}\right) \\ &\leq 2^{n} \cdot e^{-\Omega\left(\frac{n^5}{n^4}\right)}.
\end{align*}
For large enough constant $C$ the right hand side scales as $\exp(-\Omega(n))$.

We generalize this result to the sparse SYK model. Consider a distribution of sparse SYK Hamiltonians with sparsity parameter $p$, and consider the random vector $\Vec{x}$ of size $\binom{2n}{4}$ where each entry is a random Bernoulli variable with $\mathbb{E}\left[\Vec{x}_{i}\right]=p$. From a Chernoff bound (\cref{eqn:simple-binomial-chernoff}):
\begin{equation}
    \Pr\left(|\Vec{x}|_{0} \geq (1+\delta)\binom{2n}{4}p\right) \leq e^{-\delta^2\binom{2n}{4}p/2}
\end{equation}
for all $0\leq\delta\leq 1$. Thus we have that $\Pr\left(|\Vec{x}|_{0} \geq 2\E[|\Vec{x}|_{0}]\right)$ with probability $\exp(-\Omega(n))$. Conditioned on the event $|\Vec{x}|_{0} < 2\E[|\Vec{x}|_{0}]$, we have
\begin{align*}
    \Pr\left(\lambda_{\mathrm{max}}(H_{SSYK}) \geq C\sqrt{n}\right) &\leq  \Pr\left(\lambda_{\mathrm{max}}\left(\sum_{I\subseteq \mathrm{supp}(\Vec{x})} J_I \gamma_I\right) \geq Cn^{5/2}p^{1/2}\right) \\ &\leq 2^{n} \cdot e^{-\Omega\left(\frac{pn^5}{pn^4}\right)}.
\end{align*}
For large enough constant $C$ the right hand side scales as $\exp(-\Omega(n))$. Using a union bound, we can upper-bound the probability of the event that, (1) either the number of non-zero entries in our binomial vector is greater than twice the expected value, or (2) that the maximum eigenvalue is greater than $C\sqrt{n}$ given the number of non-zero entries in our binomial vector is less than twice the expected value, as less than $\exp(-\Omega(n))$. Formally, $\Pr\left(\lambda_{\mathrm{max}}(H_{SSYK}) \geq C\sqrt{n}\right)$ is upper-bounded by
\begin{align*}
     &\Pr\left(\left(\lambda_{\mathrm{max}}(H_{SSYK})\geq C\sqrt{n} \;\Big|\; |\Vec{x}|_{0} < 2\E[|\Vec{x}|_{0}]\right) \vee \left(|\Vec{x}|_{0} \geq 2\E[|\Vec{x}|_{0}\right)\right)\\&\leq \Pr\left(\lambda_{\mathrm{max}}(H_{SSYK})\geq C\sqrt{n} \;\Big|\; |\Vec{x}|_{0} < 2\E[|\Vec{x}|_{0}]\right) + \Pr\left(|\Vec{x}|_{0} \geq 2\E[|\Vec{x}|_{0}\right)
    \\&\leq \exp\left(-\Omega(n)\right) + \exp\left(-\Omega(n)\right) \leq \exp\left(-\Omega(n)\right).
\end{align*}

Thus we see that the maximum eigenvalue of the sparse SYK value is $O(\sqrt{n})$ for all $p$ with probability exponentially close to 1.
\section{Energy upper bounds on low-magic fermionic ansatz}\label[appendix]{appendix:entangled-gaussian}
The tensor network construction of \cite{hastings2023fieldtheory} was originally devised to prove upper bounds on the energy achievable by two different ansatz for the SYK model. However the proofs apply to any Hamiltonians which satisfy the norm bound of $O(k)^{2k}$ given in \cref{lemma:k-moment-syk-bound} when $p\geq\Omega(\log n/n^2)$.

\paragraph{Superpositions of Gaussian states.}

$\ket{\psi}=\sum_{j=1}^k \alpha_j \ket{\psi_j}$, where $\ket{\psi_j}$ are each (pure) fermionic Gaussian states and $\ket{\psi}$ is properly normalized. Hastings \cite{hastings2023fieldtheory} shows that the following is a corollary of \cref{lemma:k-moment-syk-bound}: If $\log\left(\sum_{j=1}^k \alpha_j\right) \leq o(n^{1/4})$, then $\bra{\psi}H\ket{\psi}\leq o(\sqrt{n})$. If the Gaussian states are all orthogonal and $k \leq \exp(o(n^{1/4}))$, we get that $\bra{\psi}H\ket{\psi}\leq o(\sqrt{n})$.

\paragraph{Lanczos algorithm.}

The Lanczos algorithm is an iterative method of computing the top eigenvectors/eigenvalues of Hermitian matrices \cite{lanczos1950iteration,paige1972computational}. It has been successfully applied to finding approximate ground states in quantum many-body physics \cite{weikert1996block,Chen_2011} by computing the vectors $H\ket{\psi}$, $H^2\ket{\psi}$, $\dots, H^k\ket{\psi}$ for a given Hamiltonian $H$, a given a starting state $\ket{\psi}$, and some finite number of iterations $k$. These vectors define a subspace called the Krylov subspace, and the ground state of the effective Hamiltonian $H$ on the Krylov subspace can be computed.

Hastings \cite{hastings2023fieldtheory} shows that \cref{lemma:k-moment-syk-bound} implies the following: if $\ket{\psi}$ is produced from the Lanczos algorithm with $k \leq o(n^{1/4})/\log n$ iterations with a Gaussian state as the initial starting state, then $\bra{\psi}H\ket{\psi}\leq o(\sqrt{n})$.
\section{Moment bound proofs}\label[appendix]{sec:moment-proofs}
We provide the proof of \cref{lemma:moment_bound} here:

\begin{proof}[Proof of \cref{lemma:moment_bound}]
We split the value of $A(1)=D_0 + D_1 + D_2 + D_3 + D_4$, where the terms have index sets satisfying:
\begin{itemize}
    \item $D_0$: $\left(S,j\right)\neq\left(S',k\right)\neq\left(S'',\ell\right)\neq\left(S,j\right)$.
    \item $D_1$: $\left(S,j\right)=\left(S',k\right)=\left(S'',\ell\right)$.
    \item $D_2$: $\left(S',k\right)=\left(S'',\ell\right)\neq\left(S,j\right)$.
    \item $D_3$: $\left(S,j\right)=\left(S'',\ell\right)\neq\left(S',k\right)$.
    \item $D_4$: $\left(S,j\right)=\left(S',k\right)\neq\left(S'',\ell\right)$.
\end{itemize}
We upper bound each term separately with $C\sqrt{n}$. By the triangle inequality, this upper bounds $\norm{A(1)}{2n}$ up to a constant. Throughout, we set $D_{i, \text{norm}} \triangleq D_i / \frac{Cp^{-3/2}}{n^5}$ for all $i=0,1,2,3,4$. 

We also utilize the following decoupling property of polynomials of independent random variables $f(Y_1, Y_2,\ldots,Y_n)$:
\begin{proposition}[Decoupling of random polynomials \cite{decoupled_gaussian}]\label{prop:decoup_lem}
    \begin{equation}
    \Pr\left[\left| \sum_{1 \leq i_1 \neq \ldots \neq i_k \leq n} f_{i_1\ldots i_k}(Y_{i_1},\ldots,Y_{i_k})\right| \geq t  \right] \leq C_k\Pr\left[\left| \sum_{1 \leq i_1 \neq \ldots \neq i_k \leq n} f_{i_1\ldots i_k}(Y_{i_1}^{(1)},\ldots,Y_{i_k}^{(k)})\right| \geq t  \right]
    \end{equation}
where $C_k$ is a constant depending on $k$ and each $\{Y_i^{(1)}\},\dots, \{Y_i^{(k)}\}$ are $k$ independent copies of the set $\{Y_i\}$ of random variables.
\end{proposition}
Previous works on optimizing SYK \cite{hastings2023optimizing, Herasymenko_2023} used this decoupling property when the $Y_i$ were Gaussian random variables. However, the property more generally applies when the $X_i$ are drawn independently from any distribution~\cite{decoupled_gaussian}.

\subsection{Bounding \texorpdfstring{$D_0$}{D0}}
Using the decoupling proposition \cref{prop:decoup_lem}, we can reduce to bounding the norm of:
\begin{equation}\label{eqn:ho-index-notation}
    \tilde{D}_0 \triangleq \frac{Cp^{-3/2}}{n^5} \sum_{r,s,t,u,w,y,z,a,b,d} (X_{rsta} J_{rsta})^{(1)} (X_{utwb} J_{utwb})^{(2)} (X_{uyzd} J_{uyzd})^{(3)} ,
\end{equation}
where the $(\cdot)^{\left(i\right)}$ superscripts denote independent (``decoupled'') random variables, the $X$ variables are independent $\Bernoulli{p}$ random variables, and the summation uses the index notation from \cite{hastings2023optimizing}: $r,s,t,u,w\in\left[n_1\right]$ and $a,b,d\in\left[n_2\right]$. In the above, the superscripts indicate that the random variables are independent if the superscript index is not equal. This expression is obtained by expanding the $S = (S_1, S_2, S_3)$ sets into individual indices, because $f(S,S',S'',j,k,l) = 1$ if $(|S''\cap S| \text{ is odd})\wedge(|S'\cap (S''\triangle S)|+\delta_{k,l} \text{ is odd})$ forces at least two pairs of the 12 indices to be identical in the non-zero terms of the summation. Thus, due to symmetry, we can arbitrarily fix two pairs of indices to be equivalent ($u$ and $t$ respectively) and sum over all remaining indices, and the set of terms in $D_0$ in our actual expression will be a subset of the summation in \cref{eqn:ho-index-notation}. This will only increase the $2n^{th}$ moment.

Let $B(f(n), p)$ denote independent binomial random variables $\Binomial{f(n), p}$, where $f(n)$ is some integer-valued polynomial of $n$. We have:
\begin{align*}
    \tilde{D}_0 &= \frac{Cp^{-3/2}}{n^5} \sum_{t, u} \sum_{rsa} (X_{rsta} J_{rsta})^{(1)} \sum_{wb} (X_{utwb} J_{utwb})^{(2)} \sum_{yzd} (X_{uyzd} J_{uyzd})^{(3)}  \\
    &= \frac{Cp^{-3/2}}{n^5} \sum_{tu} \sum_{rsa} (X_{rsta} J_{rsta})^{(1)} \sum_{wb} (X_{utwb} J_{utwb})^{(2)} \cdot B^{(3)}_u(C_3 n^3, p)^{1/2} \cdot L_u \\
    &= \frac{Cp^{-3/2}}{n^5} \sum_{tu} \sum_{rsa} (X_{rsta} J_{rsta})^{(1)} \cdot B^{(2)}_{tu}(C_2 n^2, p)^{1/2} \cdot K_{tu} \cdot B^{(3)}_u(C_3 n^3, p)^{1/2} \cdot L_u \\
    &= \frac{Cp^{-3/2}}{n^5} \sum_{tu} B^{(1)}_t(C_1 n^3, p)^{1/2} \cdot H_t \cdot B^{(2)}_{tu}(C_2 n^2, p)^{1/2}  \cdot K_{tu} \cdot B^{(3)}_{u}(C_3 n^3, p)^{1/2} \cdot L_u \\
    &= \frac{Cp^{-3/2}}{n^5} \sum_{tu} \left(B^{(1)}_t(C_1 n^3, p) \cdot B^{(2)}_{tu}(C_2 n^2, p) \cdot B^{(3)}_u(C_3 n^3, p)\right)^{1/2}  H_t K_{tu} L_u,
\end{align*}
where $H_t$, $K_{t,u}$, and $L_u$ are independent standard Gaussians; $B_t^{(1)}, B_{tu}^{(2)}, B_u^{(3)}$ denote independent random binomial variables; and $C_1,C_2,C_3>0$ are constants.

Let $B_{\text{small}}$ denote the event: $\forall t,u \ B^{(1)}_t(C_1 n^3, p) \leq (C_1+c)pn^3 \wedge B^{(3)}_u(C_3 n^3, p) \leq (C_3+c)pn^3$ for some sufficiently large constant $c$, and denote $B_{\text{large}}$ as its complement. We have:
\begin{equation}
    \norm{D_0}{2n} \leq \left(\E\left[\tilde{D}_0^{2n} \mid B_{\text{small}}\right] + \Pr(B_{\text{large}})\E\left[\tilde{D}_0^{2n} \mid B_{\text{large}}\right]\right)^{1/2n}.
\end{equation}

\paragraph{First term.}
We can directly use \cite{hastings2023optimizing} to bound the first term. As a reminder, we set $\tilde{D}_{0, \text{norm}} \triangleq \tilde{D}_0 / \frac{Cp^{-3/2}}{n^5}$. We have:
\begin{align*}
    \E[\tilde{D}^{2n}_{0, \text{norm}} \mid B_{\text{small}}] &\leq \left((C_1+c)pn^3 \cdot (C_3+c)pn^3\right)^{2n} \cdot \E\left[\left(\sum_{tu} B^{(2)}_{tu}(n^2,p)^{1/2} H_t K_{tu} L_u\right)^{2n}\right] \\
    &\leq (Cpn^3)^{2n} \cdot \mathlarger{\sum}_{b=1}^{n^2} \E\left[\left(\sum_{tu}b^{1/2} H_t K_{tu} L_u\right)^{2n}\right]\Pr(\max_{tu}B^{(2)}_{tu} = b) \\
    &\leq (Cpn^3)^{2n}  n^{3n} \cdot\mathlarger{\sum}_{b=1}^{n^2} b^{n} \Pr(\max_{tu}B^{(2)}_{tu} = b),
\end{align*}
where $\norm{\sum_{tu} H_t K_{tu} L_u}{2n}^{2n} \leq O(n^{3/2})$ is a bound due to \cite{hastings2023optimizing}. By the union bound, 
\begin{equation}
    \sum_{b=1}^{n^2} b^{n} \Pr(\max_{tu}B^{(2)}_{tu} = b)\leq Cn^2\sum_{b=1}^{n^2} b^{n} \Pr(B^{(2)}_{tu} = b)
\end{equation}
for some fixed $t,u$; the sum is just the $n^{th}$-moment of a $\Binomial{C_2 n^2, p}$ distribution. Using the general upper bound of the moments of $X \sim \Binomial{m, q}$:
\begin{equation}
    \E[X^c] \leq (mq)^c \exp{\left(\frac{c^2}{2mq}\right)},
\end{equation}
we get:
\begin{align*}
    (Cpn^3)^{2n} n^{3n} \cdot \mathlarger{\sum}_{b=1}^{n^2} b^{n} \Pr(\max_{tu}B^{(2)}_{tu} = b) &\leq (Cpn^3)^{2n}  n^{3n+2} \left(Cn^2p\right)^{n} \exp{\left(\frac{Cn^2}{2n^2p}\right)} \\
    &\leq (Cpn^3)^{2n}  n^{5n+2} p^{n} \exp{(C/2p)}\\
    &\leq (Cn^{11/2} p^{3/2})^{2n} \exp{(C/p)},
\end{align*}
with the last inequality holding for sufficiently large $n$. Since we wish to upper bound $\E[D^{2n}_{0, \text{norm}} \mid B_{\text{small}}]$ by $(C n^{11/2+O(1/n)} p^{3/2})^{2n}$, we must have $\exp{(1/p)} \leq C^{2n}$, or equivalently $p \geq \Omega(1/n)$.

\paragraph{Second term.}
Let $C_i^+\triangleq C_i+c$ for $i=1,2,3$. By a union bound and Chernoff bound on the binomial random variables (\cref{eqn:simple-binomial-chernoff}), $\Pr(B_{\text{large}})$ is bounded by
\begin{align*}
     &\Pr\left(\bigcup_{t,u}\left(B_t(C_1 n^3, p) > C_1^+pn^3\right)\lor \left(B_{tu}(C_2 n^2) > C_2^+pn^2\right) \lor \left(B_u(C_3 n^3, p) > C_3^+pn^3\right)\right)\\
    &\leq \mathlarger{\sum}_{t,u} \Pr\left(B_t(C_1 n^3, p) > C_1^+n^3\right) + \Pr\left(B_{tu}(C_2 n^2, p) > C_2^+pn^2\right) + \Pr\left(B_u(C_3 n^3, p) > C_3^+pn^3\right) \\
    &\leq Cn\exp\left(-\left(\frac{c^2}{c+2}\right)pn^3\right) + Cn^2\exp\left(-\left(\frac{c^2}{c+2}\right)pn^2\right) + Cn\exp\left(-\left(\frac{c^2}{c+2}\right)pn^3\right) \\
    &\leq Cn^2\exp\left(-\Omega(pn^2)\right).
\end{align*}
Thus
\begin{equation}
    \norm{D_0}{2n} \leq \left(\E\left[\tilde{D}_0^{2n} \mid B_{\text{small}}\right] + n^2\exp\left(-\Omega(pn^2)\right)\E\left[\tilde{D}_0^{2n} \mid B_{\text{large}}\right]\right)^{1/2n}
\end{equation}
Considering the worst case where each binomial random variable $X$ appearing in our expressions is $1$, we have that:
\begin{equation}
    \E\left[\tilde{D}_{0, \text{norm}}^{2n} \mid B_{\text{large}}\right] \leq (n^4)^{2n} \left(\sum_{tu} H_t K_{tu} L_u\right)^{2n} \leq n^{11n}.
\end{equation}
Taking into consideration the probability factor,
\begin{equation}
    n^2\exp(-\Omega(pn^2)) \E[\tilde{D}_{0, \text{norm}}^{2n} \mid B_{\text{large}}] \leq \exp(-\Omega(pn^2))\cdot n^{11n+2}.
\end{equation}
Having $p \geq \Omega(\frac{\log n}{n})$ and our constant factor $c$ large enough, we can make $\exp(-\Omega(pn^2)) = \exp(-12n\log n) = n^{-12n}$.

Putting the main term and tail term together, we have
\begin{equation}
    \norm{D_0}{2n} \leq \frac{Cp^{-3/2}}{n^5} \left(C^{2n} n^{11n}p^{3n} + C/n^{n-2}\right)^{1/2n} \leq C\sqrt{n}
\end{equation}
for sufficiently large $n$, as desired.

\subsection{Bounding \texorpdfstring{$D_1$}{D1}}
Recall:
\begin{equation}
    D_1 = C\frac{1}{\sqrt{n}}\binom{n}{3}^{-3/2}p^{-3/2} \sum_{j=1}^{n_2} \sum_{\substack{S\subset[n_1]\\|S|=3}} X_{S, j} J_{S,j}^3.
\end{equation}

To bound this, we can use hypercontractivity (Bonami's lemma \cite{Bonami1970}). $D_1$ is a degree-3 Gaussian polynomial with $B = \Binomial{O(n^4), p}$ terms. Let $B_{\text{small}}$ denote the event $B \leq (1+c)pn^4$ for some constant $c$, and $B_{\text{large}}$ as its complement. We have:
\begin{equation}
    \norm{D_1}{2n} \leq (2n-1)^{3/2} \norm{D_1}{2}
    \leq (2n-1)^{3/2} \left(\E[D_1^2 \mid B_{\text{small}}] + \Pr(B_{\text{large}}) \E[D_1^2 \mid B_{\text{large}}]\right)^{1/2}.
\end{equation}

\paragraph{First term.}
\begin{align*}
    \frac{n\binom{n}{3}^3 p^3}{C}\E[D_{1, \text{norm}}^2 \mid B_{\text{small}}]
    &\leq \E\left[\left(\sum_{i=1}^{(1+c)pn^4} J_i^3\right)^2\right] \\
    &= \E\left[\sum_{i=1}^{(1+c)pn^4} J_{i}^6\right] \\
    &\leq O(pn^4) \\
\end{align*}

\paragraph{Second term.}
By a Chernoff bound (\cref{eqn:simple-binomial-chernoff}) on the binomial random variable $B$,
\begin{align*}
    \Pr(B_{\text{large}}) &= \Pr\left(B > (1+c)pn^4\right) \\
    &\leq \exp\left(-\left(\frac{c^2}{c+2}\right)pn^4\right) \\
    &\leq \exp(-\Omega(pn^4)).
\end{align*}
As $B\leq n^4$ always, we have:
\begin{equation}
    \frac{n\binom{n}{3}^3 p^3}{C}\E\left[D_{1, \text{norm}}^2 \mid B_{\text{large}}\right] \leq \E\left[\left(\sum_{i=1}^{n^4} J_{i}^3\right)^2\right] \leq Cn^4.
\end{equation}
Putting the main term and tail term together, we have:
\begin{equation}
    \norm{D_1}{2n} \leq \frac{Cp^{-3/2}}{n^5} n^{3/2} \left(pn^4 + \exp(-\Omega(pn^4)\right)\cdot Cn^4)^{1/2} \leq O(p^{-1}n^{-3/2}).
\end{equation}
Having $p \geq \Omega(1/n^2)$ gives $\norm{D_1}{2n} \leq O(\sqrt{n})$ as desired.

\subsection{Bounding \texorpdfstring{$D_2, D_3, D_4$}{D2, D3, D4}}
Taking into account the index constraints from $f$, we have:
\begin{equation}
    D_2=C\frac{1}{\sqrt{n}}\binom{n}{3}^{-3/2}p^{-3/2} \sum_{j=1}^{n_2}\sum_{\substack{S,S'\subset[n_1]\\|S|=|S'|=3\\ S\neq S'}} X_{S,j}J_{S,j}^2 X_{S', j} J_{S', j}.
\end{equation}

After decoupling our Gaussian random variables, up to relabelings the norms of $D_2, D_3$ and $D_4$ terms are bounded in the same way when considering the SYK model with $q=4$, so we focus on $D_2$ without loss of generality.

Define $\{X'_{S,j}J'_{S,j}\}$ as a new independent set of random variables where the $X_{S,j}'$ are i.i.d.\ Bernoulli and the $J_{S,j}'$ are i.i.d.\ Gaussian. By Proposition~\ref{prop:decoup_lem}, we need only bound the norm of:
\begin{align*}
    \left\lVert\tilde{D}_2\right\rVert_{2n} &=  \left\lVert C\frac{1}{\sqrt{n}}\binom{n}{3}^{-3/2}p^{-3/2} \sum_{j, S\neq S'} X_{S,j}J_{S,j}^2 X'_{S', j} J'_{S', j}\right\rVert_{2n}\\
    &\leq\left\lVert Cn^{-5}p^{-3/2} \sum_{j=1}^{n_2} \sum_{\substack{S' \subset[n_1]\\|S'|=3}} X'_{S', j} J'_{S', j} \sum_{\substack{S \subset[n_1]\\|S|=3}} X_{S,j} J_{S,j}^2\right\rVert_{2n}\\
    &\leq\left\lVert Cn^{-5}p^{-3/2} \sum_{j=1}^{n_2} \chi_j \Big(\sum_{\substack{S' \subset[n_1]\\|S'|=3}} X'_{S', j} J'_{S', j}\Big)\right\rVert_{2n},
\end{align*}
where $\chi_j$ is a chi-squared random variable with $B_{1,j} = \Binomial{n^3, p}$ degrees of freedom. Conditioned on the $X_{S',j}$, $\sum_{S'} X_{S', j} J_{S', j}$ is also a sum of $B_{2,j}$ standard Gaussians, where $B_{2,j} = \Binomial{n^3, p}$. This culminates in:
\begin{equation}
    \norm{D_2}{2n} \leq Cn^{-5}p^{-3/2} \norm{\sum_{j=1}^{n_2} \chi_j B_{2,j}^{1/2} J''_j}{2n},
\end{equation}
where the $J_j''$ variables are i.i.d.\ standard Gaussians. We wish to show this value is bounded above by $C\sqrt{n}$, or equivalently that:
\begin{equation}
  \norm{\sum_{j=1}^{n_2} \chi_j B_{2,j}^{1/2} J''_j}{2n}^2 \triangleq \norm{\tilde{D}_{2, \text{norm}}}{2n}^2 \leq Cn^{11}p^{3}.
\end{equation}

Let $B_{\text{small}}$ denote the event $\left(\forall j \in [n_2], S, S' \ B_{1,j} \leq (1+c)pn^3 \wedge B_{2,j} \leq (1+c)pn^3\right)$ for some constant $c$, and $B_{\text{large}}$ as its complement. We have:
\begin{equation}
    \norm{D_{2, \text{norm}}}{2n} \leq \left(\E[\tilde{D}_{2, \text{norm}}^{2n} \mid B_{\text{small}}] + \Pr(B_{\text{large}})\E[\tilde{D}_{2, \text{norm}}^{2n} \mid B_{\text{large}}]\right)^{1/2n}.
\end{equation}

\paragraph{First term.}
We have:
\begin{equation}
    \E[\tilde{D}_{2, \text{norm}}^{2n} \mid B_{\text{small}}] = \norm{\sum_{j=1}^{n_2} \chi_j B_{2,j}^{1/2} J''_j}{2n}^{2n} \leq C^{2n}p^{n}n^{3n} \norm{\sum_{j=1}^{n_2} \chi_j J_j''}{2n}^{2n}.
\end{equation}

If we fix $\{\chi_j\}$, the expression $\sum_{j=1}^{n_2} \chi_j J_j'$ is a Gaussian with variance $\nu^2 \triangleq \sum_{j=1}^{n_2} \chi_j^2$ and $2n^{th}$-moment bounded by $C^n\nu^{2n}\cdot (2n)^n$. Thus it is sufficient to show that
\begin{equation}\label{eqn:nu_expectation}
        \E[\nu^{2n}] \leq O(p^2n^{7})^{n}
\end{equation}
to show that $\E[\tilde{D}_{2, \text{norm}}^{2n} \mid B_{\text{small}}] \leq C^{2n}n^{11n}p^{3n}$. The Laurent--Massart \cite{laurent-massart} lower-bound for chi-squared random variables gives:
\begin{equation}
    \Pr(\chi - k \geq 2\sqrt{kx} + 2x) \leq \exp{(-x)},
\end{equation}
where $k$ is the degrees of freedom of $\chi$. Plugging in $k=\left(1+c\right)pn^3$ and $x=cpn^3$ with $c\geq 1$ gives:
\begin{equation}
    \Pr(\chi_j \geq 8cpn^3) \leq \exp{(-cpn^3)}.
\end{equation}
Applying the union bound, we get (for sufficiently large $n$)
\begin{align*}
    \Pr(\nu \geq 8cpn^{7/2})
    &\leq \Pr\left(\bigcup_{j=1}^{n} \chi_j \geq 8cpn^3\right) \\
    &\leq \sum_{j=1}^n \Pr\left(\chi_j \geq 8cpn^3\right) \\
    &\leq n\cdot \exp{(-cpn^3)} \\
    &\leq \exp{\left(-\frac{1}{2}pn^3\right)}^c  \qquad \forall c\geq 1.
\end{align*}
Setting $\nu' \triangleq \nu / (8n^{7/2})$, we have
\begin{equation}
    \Pr(\nu' \geq cp) \leq \exp{\left(-\frac{1}{2}pn^3\right)}^c \qquad \forall c\geq 1.
\end{equation}
Finally to show the correctness of \cref{eqn:nu_expectation},
\begin{equation}
    \E[(\nu')^{2n}] \leq 1 + \sum_{j=1}^\infty 2^{(j+1)2n} \cdot \Pr\left(2^j \leq \nu' \leq 2^{j+1}\right) \leq 1 + \sum_{j=1}^\infty 2^{(j+1)2n} \cdot \exp{\left(-\frac{1}{2} pn^3\right)^{2^j}}.
\end{equation}
Assuming $p \geq C/n^2$, we have:
\begin{align*}
    \E[(\nu')^{2n}] &\leq 1 + \sum_{j=1}^\infty 2^{(j+1)2n} \cdot \exp{\left(-\frac{1}{2}  pn^3\right)}^{2^j} \\
    & \leq O(1),
\end{align*}
proving $\E[v^{2n}] \leq O(p^2n^{7})^n$ for large enough constants. This gives us $E[\tilde{D}_{2, \text{norm}}^{2n} \mid B_{\text{small}}] \leq C^{2n}p^{3n}n^{11n}$.

\paragraph{Second term.}
By a union bound and Chernoff bound (\cref{eqn:simple-binomial-chernoff}) on the binomial random variables,
\begin{align*}
    \Pr(B_{\text{large}}) &= \Pr\left(\bigcup_{j, S, S'} \left(B_{1, j} > (1+c)pn^3\right) \cup \left(B_{2,j} > (1+c)pn^3\right)\right) \\
    &\leq \sum_{j, S, S'} \Pr\left(B_{1, j} > (1+c)pn^3\right) + \Pr\left(B_{2,j} > (1+c)pn^3\right)\\
    &\leq n^7 \exp\left(-\left(\frac{c^2}{c+2}\right)pn^3\right) + n^7\exp\left(-\left(\frac{c^2}{c+2}\right)pn^3\right) \\
    &\leq n^7\exp(-\Omega(pn^3))
\end{align*}
Assume every binomial variable $B$ takes its maximum value. From the dense case~\cite{hastings2023optimizing}, we have:
\begin{equation}
    \E[\tilde{D}_{2, \text{norm}}^{2n} \mid B_{\text{large}}] \leq n^{11n}.
\end{equation}
Adding the probability factor,
\begin{equation}
    n^7\exp(-\Omega(pn^3)) \E[\tilde{D}_{2, \text{norm}}^{2n} \mid B_{\text{large}}] \leq \exp(-\Omega(pn^3))\cdot n^{11n + 7}.
\end{equation}
When $p \geq \Omega(\frac{\log n}{n^2})$ and our constant factor $c$ large enough, we can make the $\exp(-\Omega(pn^3))$ equal to $\exp(-12n\log n) = n^{-12n}$.

Putting the main term and tail term together, we have
\begin{equation}
    \norm{D_2}{2n} \leq \frac{Cp^{-3/2}}{n^5} (C^{2n} p^{3n}n^{11n} + C/n^{n-7})^{1/2n} \leq C\sqrt{n}
\end{equation}
for sufficiently large $n$, as desired.

\end{proof}

\end{document}